\documentclass[11pt]{article}

\usepackage{amssymb}
\usepackage{times}
\usepackage{latexsym, amssymb}
\usepackage{amsmath}
\usepackage{epic}
\usepackage{color}
\usepackage{fancyhdr}
\usepackage{url}
\usepackage{fullpage}
\usepackage{graphicx}
\usepackage[ruled, vlined]{algorithm2e} 
\usepackage{authblk}

\newenvironment{proof}{\noindent\textbf{Proof: }\ignorespaces}{}
\newenvironment{proofof}[1]{\noindent\textbf{Proof of #1: }\ignorespaces}{}
\newcommand{\qed}{\hspace*{\fill}$\Box$\medskip}

\newcommand{\topic}[1]{\vspace{4pt} \noindent {\bf #1}}

\newcommand{\probA}{\mathfrak{A}}
\newcommand{\classA}{\mathfrak{C}_\mathsf{bounded}}
\newcommand{\classB}{\mathfrak{C}_\mathsf{concave}}
\newcommand{\classC}{\mathfrak{C}_\mathsf{increasing}}
\newcommand{\realization}{\mathbf{r}}

\newcommand{\calF}{\mathcal{F}}

\newcommand{\bV}{\mathbb{V}}

\newcommand{\Prob}{\Pr}

\newcommand{\bfv}{\mathbf{v}}

\newcommand{\tS}{\widetilde{S}}

\newcommand{\bw}{{\bf w}}

\newcommand{\barmu}{\widehat{\mu}}
\newcommand{\hnu}{\widehat{\nu}}

\newcommand{\opt}{\mathsf{OPT}}

\newcommand{\vsmall}{\mathsf{nm}}

\newcommand{\hugevalue}{\mathsf{Huge}}
\newcommand{\Hg}{\mathsf{Hg}}
\newcommand{\val}{\mathsf{Val}}

\newcommand{\lb}{\mathsf{L}}
\newcommand{\ub}{\mathsf{U}}

\newcommand{\Exp}{\mathbb{E}}

\newcommand{\dist}{\mathsf{dis}}
\newcommand{\len}{w}

\newcommand{\eat}[1]{}

\newcommand{\uti}{\mu}
\newcommand{\umax}{\|\mu\|_\infty}

\newcommand{\tmu}{\widetilde{\mu}}
\newcommand{\ttau}{\widetilde{\tau}}

\newcommand{\bnu}{\overline{\nu}}
\newcommand{\tnu}{\widetilde{\nu}}

\renewcommand{\d}{\mathrm{d}}
\newcommand{\poly}{\mathrm{poly}}

\newcommand{\eum}{\mathsf{EUM}}
\newcommand{\eummax}{$\mathsf{EumMax}$}
\newcommand{\STOCHSP}{$\mathsf{Stoch}$-$\mathsf{SP}$}

\newcommand{\utidecomp}{\textsc{ExpSum-Approx}}

\newcommand{\tchi}{\widetilde{\chi}}
\newcommand{\algo}{\textsc{Fourier}}

\newcommand{\feature}{\mathsf{Ft}}
\newcommand{\conf}{\mathsf{Cf}}

\long\def\ignore#1{}
\newtheorem{theorem}{Theorem}

\newtheorem{lemma}{Lemma}

\newtheorem{corollary}{Corollary}

\newtheorem{example}{Example}

\title{Maximizing Expected Utility for Stochastic Combinatorial Optimization Problems
\footnote{
A preliminary version of the paper appeared 
in the Proceedings of the 52nd Annual IEEE Symposium on Foundations of Computer Science
(FOCS), 2011.
}
}

\author[1]{Jian Li\thanks{lijian83@mail.tsinghua.edu.cn}}
\author[2]{Amol Deshpande\thanks{amol@cs.umd.edu}}
\affil[1]{Institute for Interdisciplinary Information Sciences,
Tsinghua University, Beijing, P.R.China
}
\affil[2]{Department of Computer Science,
University of Maryland, College Park, USA
}

\date{}

\begin{document}

\maketitle
\begin{abstract}
We study the stochastic versions of a broad class of combinatorial problems
where the weights of the elements in the input dataset
are uncertain. The class of problems that we study includes shortest paths,
minimum weight spanning trees, and minimum weight matchings,
and other combinatorial problems like
knapsack. We observe that the expected value is inadequate in capturing different
types of {\em risk-averse} or {\em risk-prone} behaviors, and instead we consider
a more general objective which is to maximize the {\em expected utility} of the solution for some given utility function,
rather than the expected weight
(expected weight becomes a special case).
Under the assumption that
 there is a pseudopolynomial time algorithm for the {\em exact} version of the problem
(This is true for the problems mentioned above),
\footnote{Following the literature~\cite{papadimitriou2000approximability}, we differentiate
between {\em exact} version and {\em deterministic} version of a problem;
in the exact version of the problem, we are given a target value and asked to find a solution (e.g., a path)
with exactly that value (i.e., path length).}
we can obtain the following approximation results for several important classes of utility functions:
\begin{enumerate}
\item
If the utility function $\uti$ is continuous, upper-bounded by a constant and 
$\lim_{x\rightarrow+\infty}\uti(x)=0$,
we show that we can obtain a polynomial time approximation algorithm
with an {\em additive error} $\epsilon$ for any constant $\epsilon>0$.
\item
If the utility function $\uti$ is a concave increasing function, 
we can obtain a polynomial time approximation scheme (PTAS).
\item
If the utility function $\uti$ is increasing and has a bounded derivative, 
we can obtain a polynomial time approximation scheme.
\end{enumerate}
Our results recover or generalize several prior results on stochastic shortest path,
stochastic spanning tree, and stochastic knapsack.
Our algorithm for utility maximization makes use of the separability of exponential utility and a technique to decompose a
general utility function into exponential utility functions, which may be useful in other stochastic optimization problems.
\end{abstract}

\section{Introduction}

The most common approach to deal with optimization problems in presence of uncertainty is to optimize
the expected value of the solution.
However, expected value is inadequate in expressing diverse people's preferences
towards decision-making under uncertain scenarios.
In particular, it fails at capturing different {\em risk-averse} or {\em risk-prone}  behaviors
that are commonly observed.
Consider the following simple example where we have two lotteries $L_{1}$ and $L_{2}$.
In $L_{1}$, the player could win $1000$ dollars with probability $1.0$,
while in $L_{2}$ the player could win $2000$ dollars with probability $0.5$ and $0$ dollars otherwise.
It is easy to see that both have the same expected payoff of $1000$ dollars.
However, many, if not most, people would treat $L_{1}$ and $L_{2}$ as two completely different choices.
Specifically, a risk-averse player is likely to choose $L_{1}$ and a risk-prone player
may prefer $L_{2}$ (Consider a gambler who would like to spend 1000 dollars to play double-or-nothing).
A more involved but also more surprising example is the
{\em St. Petersburg paradox} (see e.g.,~\cite{robert04stpetersburg}) which
has been widely used in the economics literature
as a criticism of expected value.
The paradox is named from Daniel Bernoulli's presentation of the problem,
published in 1738 in the Commentaries of the Imperial Academy of Science of Saint Petersburg.
Consider the following game: you pay a fixed fee $X$ to enter the game.
In the game, a fair coin is tossed repeatedly until a tail appears ending the game.
The payoff of the game is $2^{k}$ where $k$ is the number of heads that appear, i.e.,
you win $1$ dollar if a tail appears on the first toss, 2 dollars if a head appears on the first toss and a tail on the second, 4 dollars if a head appears on the first two tosses and a tail on the third and so on.
The question is what would be a fair fee $X$ to enter the game?
First, it is easy to see that the expected payoff is
$$
\Exp[\text{payoff}]=\frac{1}{2}\cdot 1+\frac{1}{4}\cdot 2 + \frac{1}{8}\cdot 4 + \frac{1}{16}\cdot 8 + \cdots
    =\frac{1}{2} + \frac{1}{2} + \frac{1}{2} + \frac{1}{2} + \cdots
    =\sum_{k=1}^\infty {1 \over 2}=\infty
$$
If we use the expected payoff as a criterion for decision making,  we should therefore play the game at any finite price $X$
(no matter how large $X$ is)
since the expected payoff is always larger.
However, researchers have done extensive survey and found that not many people would pay even 25 dollars to play the game \cite{robert04stpetersburg}, which significantly deviates from what the expected value criterion predicts.
In fact, the paradox can be resolved by expected utility theory with a logarithmic utility function, suggested by Bernoulli himself~\cite{bernoulli1954exposition}.
We refer interested reader to \cite{samuelson1977st,robert04stpetersburg} for more information.
These observations and criticisms have led researchers, especially in Economics,
to study the problem from a more fundamental perspective
and to directly maximize user satisfaction, often called {\em utility}.
The uncertainty present in the problem instance naturally leads us to optimize the {\em expected utility}.

Let $\calF$ be the set of feasible solutions to an optimization problem.
Each solution $S\in \calF$ is associated with a random weight $w(S)$.
For instance, $\calF$ could be a set of lotteries and $w(S)$ is the (random) payoff of lottery $S$.
We model the risk awareness of a user by a utility function $\uti: \mathbb{R} \rightarrow \mathbb{R}$:
the user obtains $\uti(x)$ units of utility
if the outcome is $x$, i.e., $w(S)=x$.
Formally, the {\em expected utility maximization principle} is simply stated as follows:
the most desirable solution $S$ is the one that maximizes
the expected utility, i.e.,
$$
S=\arg\max_{S'\in \calF}\Exp[\uti(w(S'))]
$$
Indeed, expected utility theory is a branch of utility theory
that studies ``betting preferences" of people with regard to uncertain outcomes (gambles).
The theory was formally initiated by von Neumann and Morgenstern
in 1940s \cite{VonNeumann1947,Fishburn70}
\footnote{
Daniel Bernoulli also developed many ideas, such as risk aversion and utility, in his work
{\em Specimen theoriae novae de mensura sortis (Exposition of a New Theory on the Measurement of Risk)}
in 1738 \cite{Bernoulli1738}.
}
who gave an axiomatization of the theory
(known as {\em von Neumann-Morgenstern expected utility theorem}).
The theory is well known to be versatile in expressing diverse risk-averse or risk-prone behaviors.

In this paper, we focus on the following broad class of
combinatorial optimization problems.
The deterministic version of the problem has the following form:
we are given a ground set of elements $U=\{e_{i}\}_{i=1...n}$; each element $e$ is associated with a weight $w_{e}$;
each feasible solution is a subset of the elements satisfying some property.
Let $\calF$ denote the set of feasible solutions.
The objective for the deterministic problem is to find a feasible solution $S$ with the minimum (or maximum) total weight $w(S)=\sum_{e\in S}w_{e}$.
We can see that many combinatorial problems such as shortest path, minimum spanning tree, and minimum weight matching
belong to this class.
In the stochastic version of the problem,
the weight $w_{e}$ of each element $e$ is a nonnegative random variable.
We assume all $w_{e}$s are independent of each other.
We use $p_{e}(.)$ to denote the probability density function for $w_{e}$ (or probability mass function in the discrete case).
We are also given a utility function $\uti:\mathbb{R}^{+}\rightarrow \mathbb{R}^{+} $
which maps a weight value to a utility value.
By the expected utility maximization principle, our goal here is to find a feasible solution $S\in \calF$
that maximizes the expected utility, i.e., $\Exp[\mu(w(S))]$.
We call this problem the {\em expected utility maximization ($\eum$)} problem.

Let us use the following toy example to illustrate the rationale behind $\eum$.
There is a graph with two nodes $s$ and $t$ and two parallel links $e_{1}$ and $e_{2}$.
Edge $e_{1}$ has a fixed length $1$ while the length of $e_{2}$ is $0.9$ with probability $0.9$ and $1.9$ with probability
$0.1$ (the expected value is also $1$).
We want to choose one edge to connect $s$ and $t$.
It is not hard to imagine that a risk-averse user would choose $e_{1}$ since $e_{2}$ may turn out
to be a much larger value with a nontrivial probability. We can capture such behavior using the utility function (\ref{ex:utility1})
(defined in Section~\ref{subsec:contribution}). Similarly, we can capture the risk-prone behavior by using, for example,
the utility function $\uti(x)=\frac{1}{x+1}$. It is easy to see that $e_{1}$ maximizes the expected utility in the former case,
and $e_{2}$ in the latter.

\subsection{Our Contributions}
\label{subsec:contribution}

In order to state our contribution, we first recall some standard terminologies.
A polynomial time approximation scheme (PTAS) is an algorithm which takes an instance of a minimization problem
(a maximization problem resp.)
and a parameter $\epsilon>0$ and produces a solution whose cost  is at most $(1+\epsilon)\opt$ (at least $(1-\epsilon)\opt$ resp.),
and the running time, for any fixed constant $\epsilon>0$, is polynomial in the size of the input, where $\opt$
is the optimal solution.
We use $\probA$ to denote the deterministic combinatorial optimization problem under consideration,
and $\eum(\probA)$ the corresponding expected utility maximization 
problem.
The {\em exact version} of $\probA$ asks the question whether
there is a feasible solution of $\probA$ with weight exactly equal to a given integer $K$.
We say an algorithm runs in {\em pseudopolynomial time} for the exact version of $\probA$
if the running time is polynomial in $n$ and $K$.
For many combinatorial problems, a pseudopolynomial algorithm for the exact version
is known. Examples include shortest path, spanning tree, matching and knapsack.

We discuss in detail our results for $\eum$.
We start with a theorem which underpins our other results. 
We denote $\umax=\sup_{x\geq 0} |\uti(x)|$.
We say a function $\tmu(x)$ is an {\em $\epsilon$-approximation} of 
$\uti(x)$ if $|\tmu(x)-\uti(x)|\leq \epsilon \umax$ for all $x\geq0$.
We allow $\tmu(x)$ to be a complex function
and $|\tmu(x)|$ denote its absolute value
(as we will see shortly, $\tmu(x)$ takes the form of a finite sum of complex exponentials). 
\footnote{
	In practice, the user only needs to specify a real 
	utility function $\mu(x)$.
	The complex function $\tmu(x)$ is used to approximate 
	the real utility function $\mu(x)$.
	}

\begin{theorem}
\label{thm:mainthm}
Assume that there is a pseudopolynomial algorithm for the exact version of $\probA$.
Further assume that given any constant $\epsilon > 0$,
we can find an $\epsilon$-approximation of the utility function $\uti$ as $\tmu(x)=\sum_{k=1}^{L}c_{k}\phi_k^x$,
where 
$|\phi_{k}|\leq 1$ for all $1\leq k\leq L$ ($\phi_k$ may be complex numbers).
Let $\tau=\max_k |c_k|/\umax $.
Then, there is an algorithm that runs in time $(n\tau/\epsilon)^{O(L)}$ and finds a feasible solution $S\in \calF$ such that
$$
\Exp[\uti(w(S))]\geq \opt-\epsilon \umax.
$$
\end{theorem}

From the above theorem, we can see that if we can $\epsilon$-approximate the utility function $\uti$
by a short sum of exponentials, we can obtain good approximation algorithms for $\eum$.
In this paper, we consider three important classes of utility functions.
\begin{enumerate}
\item 
(Class $\classA$) 
Consider the deterministic problem which
$\probA$ is a minimization problem, i.e.,
we would like the cost of our solution to be as small as possible.
In the corresponding stochastic version of $\probA$,
we assume that any utility function $\uti(x)\in\classA$ is nonnegative, bounded, continuous and $\lim_{x\rightarrow \infty}\uti(x) =0$
(please see below for the detailed technical assumptions). 
The last condition captures
the fact that if the cost of solution is too large, it becomes almost useless for us.
We denote the class of such utility functions by $\classA$.
\item 
(Class $\classB$) Consider 
the deterministic problem $\probA$ which is a maximization problem.
In other words, we want the value of our solution to be as large as possible.
In the corresponding stochastic version of $\probA$,
we assume that $\uti(x)$ is a nonnegative, monotone nondecreasing and concave function.
Note that concave functions are extensively used to model risk-averse behaviors in the economics literature.
We denote the class of such utility functions by $\classB$.
\item 
(Class $\classC$) 
Consider a
deterministic maximization problem $\probA$.
In the corresponding stochastic version of $\probA$,
we assume that $\uti(x)$ is a nonnegative, differentiable and increasing function.
We assume $\frac{\d}{\d x}\uti(x)\in [\lb, \ub]$ for $x\geq 0$, where $\lb,\ub>0$ are constants.
We denote the class of such utility functions by $\classC$.
We can see that functions in $\classC$ can be concave, nonconcave, convex or nonconvex.
Convex functions are often associated with risk-prone behaviors,
while nonconvex-nonconcave utility functions have been also observed in 
various settings \cite{kahneman1979prospect, fazel2005network}.
\end{enumerate}

Now, we state in details our assumptions and results for the above classes of utility functions.

\topic{Class $\classA$:}
Since $\uti$ is bounded, by scaling, without loss of generality, 
we can assume $\umax=1$.
Since $\lim_{x\rightarrow \infty} \uti(x)=0$, for any $\epsilon>0$, 
there exist a point $T_{\epsilon}$ such that
$\uti(x)\leq \epsilon$ for $x>T_{\epsilon}$.
We assume that $T_\epsilon$ is a constant only depending on $\epsilon$.
We further assume that the continuous utility function $\uti$ satisfies
the  $\alpha$-H\"{o}lder condition, i.e., $| \mu(x) - \mu(y) | \leq C \, |x - y|^{\alpha}$, for some constant $C$
and some constant $\alpha>1/2$.
We say $f$ is {\em $C$-Lipschitz} if $f$ satisfies 1-H\"{o}lder condition with coefficient $C$.
Under the above conditions, we can prove Theorem~\ref{thm:mainholder}.

\begin{theorem}
\label{thm:mainholder}
If the utility function $\uti$ belongs to $\classA$,
then, for any $\epsilon>0$, we can obtain a function $\tmu(x)=\sum_{k=1}^L c_k \phi^x_k$,
such that 
$
|\tmu(x)-\uti(x)|\leq \epsilon,
$ for $x\geq 0$,where
$$
L=
2^{O(T_\epsilon)}\poly(1/\epsilon), \quad |c_k|\leq 2^{O(T_\epsilon)}\poly(1/\epsilon), \quad |\phi_k|\leq 1 \text{ for all }k=1,\ldots, L,
$$
\end{theorem}

To show the above theorem, we use the Fourier series technique.
However, the technique cannot be used directly since it works only for periodic functions with
bounded periodicities. In order to get a good approximation for $x\in [0, \infty)$,
we leverage the fact that $\lim_{x\rightarrow \infty}\uti(x) =0$ and develop a general framework
that uses the Fourier series decomposition as a subroutine.


\begin{figure}[t]
\begin{center}
\includegraphics[width=0.8\linewidth, height=3cm]{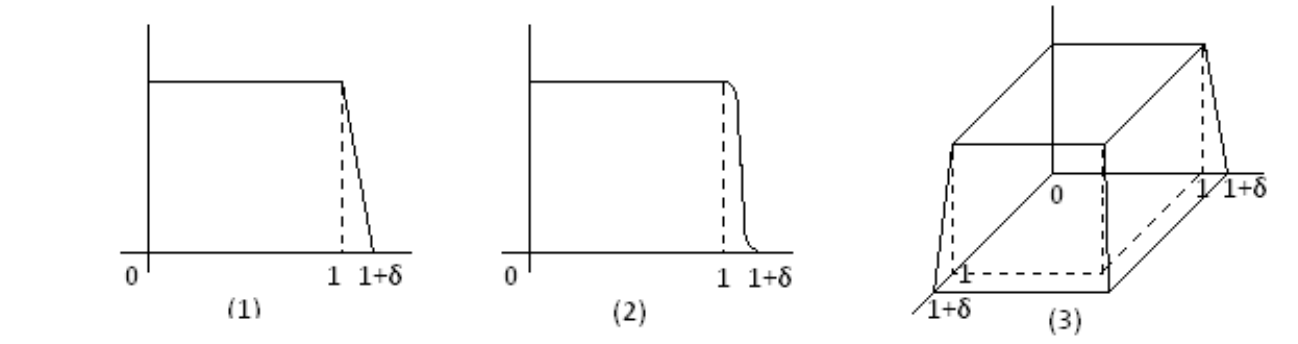}
\end{center}
\vspace{-0.9cm}
\caption{(1) The utility function $\tchi(x)$, a continuous variant of the threshold function $\chi(x)$; (2) A smoother variant of $\chi(x)$;
(3) The utility function $\tchi_{2}(x)$, a continuous variant of the 2-d threshold function $\chi_{2}(x)$.}
\vspace{-0.4cm}
\label{fig_utility}
\end{figure}

Now, we state some implications of the above results.
Consider the utility function
\begin{align}
\label{ex:utility1}
\tchi(x)= \left\{
             \begin{array}{ll}
		1 & x\in [0,1] \\
		 -\frac{x}{\delta}+\frac{1}{\delta}+1  & x\in[1, 1+\delta] \\
		 0 & x>1+\delta
	\end{array}
           \right.
\end{align}
where $\delta>0$ is a small constant (See Figure~\ref{fig_utility}(1)).
It is easy to verify that $\tchi$ is
$1/\delta$-Lipschitz and $T_\epsilon=2$ for any $\delta<1$.
Therefore, Theorem~\ref{thm:mainholder} is applicable.
This example is interesting since $\tchi$ can be viewed as a continuous variant
of the threshold function
\begin{align}
\chi(x)= \left\{
             \begin{array}{ll}
		1 & x\in [0,1] \\
		 0 & x>1
	\end{array}
           \right.,
\end{align}
 for which maximizing the expected utility is equivalent to
maximizing $\Prob(w(S) \leq 1)$.
We first note that
even the problem of computing the probability $\Prob(w(S) \leq 1)$ exactly for a fixed set $S$
is \#P-hard \cite{kleinberg1997allocating} and there is an FPTAS~\cite{li2014fully}.
Designing approximation algorithms for such special case has been considered
several times in the literature for various combinatorial problems including stochastic shortest path \cite{nikolova2006stochastic},
stochastic spanning tree~\cite{ishii1981stochastic,geetha1993stochastic},
stochastic knapsack~\cite{goel1999stochastic} and some other stochastic problems~\cite{agrawal2008stochastic,nikolova2010approximation}.

It is interesting to compare our result with
the result for the stochastic shortest path problem considered by
Nikolova et al.~\cite{nikolova2006stochastic, nikolova2010approximation}.
In \cite{nikolova2006stochastic}, they show that there is an exact $O(n^{\log n})$ time algorithm for maximizing
the probability that the length of the path is at most 1, i.e., $\Prob(w(S) \leq 1)$,
assuming all edges are normally distributed and there is a path with its mean at most $1$.
Later, Nikolova~\cite{nikolova2010approximation} extends the result to an FPTAS
for any problem under the same assumptions, if the deterministic version of the problem
has a polynomial time exact algorithm.
We can see that under such assumptions, the optimal probability is at least 
$1/2$.~\footnote{
The sum of multiple Gaussians is also a Gaussian. Hence, if we assume the mean of the length of a path (which is a Gaussian)
is at most $1$,
the probability that the length of the path is at most 1 is at least $1/2$.
}
Therefore, provided the same assumption and further assuming that $\Prob(w_{e}<0)$ is miniscule,
\footnote{
Our technique can only handle distributions with positive supports.
Thus, we have to assume that the probability that a negative value appears is miniscule (e.g., less than $1/n^2$)
and can be safely ignored
(because the probability that there is  any realized negative value is at most $1/n$).
}
our algorithm is a PTAS for maximizing $\Exp[\tchi(w(S))]$, which can be thought as a variant of the problem of maximizing
$\Exp[\chi(w(S))]$.
Indeed, we can translate this result to a bi-criterion approximation result of the following form:
for any fixed constants $\delta,\epsilon>0$, we can find in polynomial time a solution $S$ such that
$$
\Prob(w(S)\leq 1+\delta) \geq (1-\epsilon) \Prob(w(S^{*})\leq 1).
$$
where $S^{*}$ is the optimal solution (Corollary~\ref{cor:thres}).
We note that such a bi-criterion approximation was only known for exponentially distributed edges before~\cite{nikolova2006stochastic}.

Let us consider another application of our results to the stochastic knapsack problem defined in \cite{goel1999stochastic}.
Given a set $U$ of independent random variables $\{x_{1},\ldots, x_{n}\}$, with associated profits $\{v_{1},\ldots, v_{n}\}$
and an overflow probability $\gamma$, we are asked to pick a subset $S$ of $U$ such that
$$
\Prob\left(\sum_{i\in S}x_{i}\geq 1\right)\leq \gamma
$$
and the total profit $\sum_{i\in S}v_{i}$ is maximized.
Goel and Indyk \cite{goel1999stochastic} showed that, for any constant $\epsilon>0$, there is a polynomial time algorithm that can find
a solution $S$ with the profit as least the optimum and $\Prob(\sum_{i\in S}x_{i}\geq 1+\epsilon)\leq \gamma(1+\epsilon)$
for exponentially distributed variables. They also gave a quasi-polynomial time approximation scheme
for Bernoulli distributed random variables.
Quite recently, in parallel with our work, Bhalgat et al.~\cite{bhalgat10} obtained
the same result for arbitrary distributions under the assumption that $\gamma=\Theta(1)$.
Their technique is based on discretizing the distributions and is quite involved.
\footnote{
They also obtain several results related to stochastic knapsack,
using the their discretization technique, together with other ideas.
Notably, they obtained a bi-criteria PTAS for 
the {\em adaptive stochastic knapsack} problem~\cite{bhalgat10}.
}
Our result, applied to stochastic knapsack, matches that of Bhalgat et al. under the same assumption.
Our algorithm is arguably simpler and has a much better running time (Theorem~\ref{cor:knapsack}).

Equally importantly, we can extend our basic approximation scheme to handle generalizations
such as multiple utility functions and multidimensional weights.
Interesting applications of these extensions include various generalizations of stochastic knapsack, such as
{\em stochastic multiple knapsack} (Theorem~\ref{thm:multiknapsack})
and {\em stochastic multidimensional knapsack (stochastic packing)}
(Theorem~\ref{thm:multidimknapsack}).


\topic{Class $\classB$:}
We assume the utility function $\uti : [0,\infty)\rightarrow [0,\infty)$ is a concave, monotone nondecreasing function.
This is a popular class of utility functions used to model risk-averse behaviors.
For this class of utility functions, we can obtain the following theorem in Section~\ref{sec:concave}.

\begin{theorem}
\label{thm:mainthmconcave}
Assume the utility function $\uti$ belongs to $\classB$,
and there is a pseudopolynomial algorithm for the exact version of $\probA$.
Then, there is  a PTAS for $\eum(\probA)$.
\end{theorem}

Theorem~\ref{thm:mainthmconcave} is also obtained by an application of Theorem~\ref{thm:mainthm}.
However, instead of approximating the original utility function $\uti$ using a short sum of exponentials, 
which may not be possible in general,
\footnote{
Suppose $\uti$ is a finite sum of exponentials.
When $x$ approaches to infinity, either $|\uti(x)|$ is periodic, or approaches to infinity,  or approaches to 0.
}
we try approximate a truncated version of $\uti$.
Theorem~\ref{thm:mainthmconcave} recovers the recent result of \cite{bhalgat2014utility}.
Finally, we remark the technique of \cite{bhalgat2014utility} strongly relies on the concavity of $\uti$,
and seems difficult to extend to handle non-concave utility functions.

\topic{Class $\classC$:}
We assume the utility function $\uti : [0,\infty)\rightarrow [0,\infty)$ is a 
positive, differentiable, and increasing function.
For technical reasons, 
we assume $\frac{\d}{\d x}\uti(x)\in [\lb, \ub]$ for some constants $\lb,\ub>0$ and all $x\geq 0$.
For this class of utility functions, we can obtain the following theorem in Section~\ref{sec:classC}.

\begin{theorem}
\label{thm:mainthmconvex}
Assume the utility function $\uti$ belongs to $\classC$,
and there is a pseudopolynomial algorithm for the exact version of $\probA$.
Then, there is  a PTAS for $\eum(\probA)$.
\end{theorem}

Again, it may not be possible in general to approximate such an increasing function using 
a finite sum of exponentials.
Instead, we approximate a truncated version of $\uti$, similar to the concave case.
We note this is the first such result for general increasing utility functions.
Removing the bounded derivative assumption remains an interesting open problem.

We believe our technique can be used to handle other classes of utility functions
or other stochastic optimization problems.

\subsection{Related Work}

In recent years stochastic optimization problems
have drawn much attention from the computer science community and
stochastic versions of many classical combinatorial optimization problems
have been studied.
In particular, a significant portion of the efforts has been devoted to
the two-stage stochastic optimization problem.  In such a problem, in a first
stage, we are given probabilistic information about the input but the cost of selecting an
item is low; in a second stage, the actual input is revealed but the
costs for the elements are higher. We are asked to make decision after each stage
and minimize the expected cost.
Some general techniques have been developed \cite{gupta2004boosted,shmoys2006approximation}.
We refer interested reader to \cite{swamy2006approximation} for a comprehensive survey.
Another widely studied type of problems considers designing {\em adaptive} probing policies for stochastic
optimization problems where the existence or the exact weight of an element can be only known upon
a probe. There is typically a budget for the number of probes (see e.g., \cite{guha2008adaptive,cheng2008cleaning}),
or we require an irrevocable decision whether to include the probed element in the solution right after the probe
(see e.g., \cite{dean2008approximating, chen2009approximating, bansal2010lp,dean2005adaptivity,bhalgat10}).
However, most of those works focus on optimizing the expected value of the solution.
There is also sporadic work on optimizing the overflow probability or some other objectives subject to
the overflow probability constraints.
In particular, a few recent works have explicitly motivated such objectives as a way to capture
the risk-averse type of behaviors \cite{agrawal2008stochastic,nikolova2010approximation,swamy2008algorithms}.
Besides those works, there has been little work on optimizing
more general utility functions for combinatorial stochastic optimization problems
from an approximation algorithms perspective.

The most related work to ours is the stochastic shortest path problem (\STOCHSP),
which was also the initial motivation for this work.
The problem has been studied extensively for several special utility functions
in operation research community.
Sigal et al.~\cite{sigal1980stochastic} studied the problem of finding the path with greatest probability of being the shortest path.
Loui~\cite{loui1983optimal} showed that \STOCHSP\ reduces to the shortest path (and sometimes longest path) problem
if the utility function is linear or exponential.
Nikolova et al.~\cite{nikolova2006optimal} identified more specific utility and distribution combinations that
can be solved optimally in polynomial time.
Much work considered dealing with more general utility functions, such as piecewise linear or concave functions,
e.g., \cite{murthy1997exact,murthy1998stochastic,bard1991arc}. However, these algorithms are essentially heuristics and
the worst case running times are still exponential.
Nikolova et al.~\cite{nikolova2006stochastic} studied the problem of maximizing the probability that the length of the chosen path is less than
some given parameter.
Besides the result we mentioned before, they also considered Poisson and exponential distributions.
Despite much effort on this problem, no algorithm
is known to run in polynomial time
and have provable performance guarantees, especially for more general utility functions or more general distributions.
This is perhaps because the hardness comes from different sources, as also noted in \cite{nikolova2006stochastic}:
the shortest path selection per se is combinatorial; the distribution of the length of a path is
the convolution of the distributions of its edges; the objective is nonlinear; to list a few.

Kleinberg et al. \cite{kleinberg1997allocating} first considered the stochastic knapsack problem
with Bernoulli-type distributions and provided a polynomial-time $O(\log 1/\gamma)$
approximation where $\gamma$ is the given overflow probability.
In the same paper, they noticed that even computing the 
overflow probability for a fixed set of items is \#P-hard.
Li and Shi \cite{li2014fully} provided an FPTAS for computing 
the overflow probability (or the threshold probability for 
a sum of random variables).
For item sizes with exponential distributions, Goel and Indyk \cite{goel1999stochastic} provided a bi-criterion PTAS, and for
Bernoulli-distributed items they gave a quasi-polynomial approximation scheme.
Chekuri and Khanna \cite{chekuri2000ptas} pointed out that a PTAS can be obtained for the Bernoulli case using their techniques
for the multiple knapsack problem.
Goyal and Ravi \cite{goyal2009chance} showed a PTAS for Gaussian distributed sizes.
Bhalgat, Goel and Khanna \cite{bhalgat10} developed a general discretizaton technique
that reduces the distributions to a small number of equivalent classes which we can efficiently enumerate
for both adaptive and nonadaptive versions of stochastic knapsack.
They used this technique to obtain improved results for several variants of stochastic knapsack,
notably a bi-criterion PTAS for the {\em adaptive version} of the problem.
In a recent work ~\cite{li2013stoch}, 
the bi-criterion PTAS was further simplified and extended to the more general case where
the profit and size of an item can be correlated and an item can be cancelled in the middle.
Dean at al.~\cite{dean2008approximating} gave the first constant approximation for
the adaptive version of stochastic knapsack.
The adaptive version of stochastic multidimensional knapsack
(or equivalently stochastic packing) has been considered in \cite{dean2005adaptivity,bhalgat10} where
constant approximations and a bi-criterion PTAS were developed.

This work is partially inspired by our prior work on top-$k$ and other queries over probabilistic datasets~\cite{pods09_LD,li09unified}.
In fact, we can show that both the {\em consensus answers} proposed in \cite{pods09_LD}  and the {\em parameterized ranking functions}
proposed in \cite{li09unified}  follow the expected utility maximization principle where the utility functions
are materialized as distance metrics for the former and the weight functions for the latter.
Our technique for approximating the utility functions is also similar to the approximation scheme used in \cite{li09unified} in spirit.
However, no performance guarantees are provided in that work.

Recently, Li and Yuan \cite{li2013stoch} showed that an additive PTAS for $\uti\in \classA$ can be obtained
using a completely different approach, called the Poisson approximation technique.
Roughly speaking, the Poisson approximation technique allows us to extract
a constant (depending on $\epsilon$) number of features
from each distribution (called signature in \cite{li2013stoch}) and reduce the stochastic problem to
a constant dimensional deterministic optimization problem, which is similar to the algorithm presented in this paper.
We suspect that besides this superficial similarity,
there may be deeper connections between two different techniques.

There is a large volume of work on approximating functions using short exponential sums over a bounded domain, e.g.,
\cite{Osborne95modified,beylkin2002generalized,beylkin2005approximation,beylkin2010approximation}.
Some works also consider using linear combinations of Gaussians or other {\em kernels} to approximate
functions with finite support over the entire real axis
$(-\infty, +\infty)$ \cite{cheney2000}.
This is however impossible using exponentials since $\alpha^{x}$ is either periodic (if $|\alpha|=1$)
or approaches to infinity when $x\rightarrow +\infty$ or $x\rightarrow -\infty$ (if $|\alpha|\ne 1$).

\section{An Overview of Our Approach}

The high level idea of our approach is very simple and consists of the following steps:

\begin{enumerate}
\item We first observe that the problem is easy if the utility function is an exponential function.
Specifically, consider the exponential utility function $\uti(x)=\phi^{x}$ for some complex number $\phi\in \mathbb{C}$.
Fix an arbitrary solution $S$. Due to independence of the elements, we can see that
\begin{align*}
\Exp[\phi^{w(S)}]=\Exp\bigl[\phi^{\sum_{e\in S}w_e}\bigr]=\Exp\Bigl[\,\prod_{e\in S} \phi^{w_{e}}\Bigr]
=\prod_{e\in S}\Exp[\phi^{w_{e}}]
\end{align*}
Taking log on both sides, we get $\log\Exp[\phi^{w(S)}]=\sum_{e\in S}\log\Exp[\phi^{w_{e}}].$
If $\phi$ is a positive real number and  $\Exp[\phi^{w_{e}}]\leq 1$ (or equivalently, $-\log \Exp[\phi^{w_{e}}]\geq 0$), this
reduces to the deterministic optimization problem.
\item
In light of the above observation, we $\epsilon$-approximate the utility function $\uti(x)$ by a short exponential sum, i.e.,
$\sum_{i=1}^L c_{i} \phi_i^x$ with $L$ being a small value (only depending on $\epsilon$), where
($c_{i}$ and $\phi_{i}$ may be complex numbers.
Hence, $\Exp[\uti(w(S))]$ can be approximated by $\sum_{i=1}^L c_{i} \Exp[\phi_{i}^{w(S)}]$.
\item
Consider the following multi-criterion version of the problem with
$L$ objectives $\{\Exp[\phi_{i}^{w(S)}]\}_{i=1,\ldots,L}$:
given $L$ complex numbers $v_{1},\ldots, v_{L}$,
we want to find a solution $S$ such that $\Exp[\phi_{i}^{w(S)}]\approx v_{i}$ for $i=1,\ldots, L$.
We achieve this by utilizing the pseudopolynomial time algorithm for the exact version of the problem.
We argue that we only need to consider a polynomial number of
$v_{1},\ldots,v_{L}$ combinations (which we call {\em configurations}) to find out the approximate
optimum.
\end{enumerate}

In Section~\ref{sec:algorithm}, we show how to solve the multi-criterion problem provided that
a short exponential sum approximation of $\uti$ is given.
In particular, we prove Theorem~\ref{thm:mainthm}.
Then, we show how to approximate $\uti\in \classA$ by a short exponential sum by
proving Theorem~\ref{thm:mainholder} in Section~\ref{subsec:approximateutility} and Section~\ref{subsec:fourier}.
For $\uti\in \classB$ or $\uti\in \classC$, 
it may not be possible to approximate $\uti$ directly by an exponential sum, and some additional ideas are required.
The details are provided in Section~\ref{sec:concave} and Section~\ref{sec:classC}.

We still need to show how to compute $\Exp[\phi^{w_{e}}]$.
If $w_{e}$ is a discrete random variable with a polynomial size support,
we can easily compute $\Exp[\phi^{w_{e}}]$ in polynomial time.
If $w_{e}$ has an infinite discrete or continuous support,
we can not compute $\Exp[\phi^{w_{e}}]$ directly and need to approximate it.
We briefly discuss this issue and its implications in Appendix~\ref{app:approxe}.

\section{Proof of Theorem~\ref{thm:mainthm}}
\label{sec:algorithm}

Now, we prove Theorem~\ref{thm:mainthm}.
We start with some notations.
We use $|c|$ and $\arg(c)$ to denote the absolute value and the argument of the complex number $c\in \mathbb{C}$, respectively.
In other words, $c=|c|\cdot(\cos(\arg(c))+i \sin(\arg(c))) =|c| e^{i\arg(c)}$.
We always require $\arg(c)\in [0,2\pi)$ for any $c\in \mathbb{C}$.
Recall that we say the exponential sum $\sum_{i=1}^L c_{i} \phi_i^x$ is an
{\em $\epsilon$-approximation} for $\uti(x)$ if the following holds:
$$
|\uti(x)-\sum_{i=1}^L c_{i} \phi_i^x |\leq \epsilon \umax \quad\text{ for } x\geq 0.
$$



We first show that if the utility function can be decomposed exactly into a short exponential sum,
we can approximate the optimal expected utility well.

\begin{theorem}
\label{thm:algofortmu}
Assume that $\tmu(x)=\sum_{k=1}^{L}c_{k}\phi_k^x$ is the utility function
where $|\phi_{k}|\leq 1$ for $1\leq k\leq L$.
Let $\tau=\max_k |c_k|/\umax $.
We also assume that there is a pseudopolynomial algorithm for the exact version of $\probA$.
Then, for any $\epsilon>0$, there is an algorithm that runs in time $(n/\epsilon)^{O(L)}$
and finds a solution $S$ such that
$$|\Exp[\tmu(\len(S))]-\Exp[\tmu(\len(\tS))]|<\epsilon \umax,$$
where $\tS=\arg\max_{S'}|\Exp[\tmu(\len(S'))|$.
\end{theorem}

We use the scaling and rounding technique that has been used often in
multi-criterion optimization problems (e.g., \cite{safer04fully,papadimitriou2000approximability}).
Since our objective function is not additive and not monotone,
the general results for multi-criterion optimization
\cite{papadimitriou2000approximability,mittal2008general, safer04fully, ackermann2005decision}
do not directly apply here. We provide the details of the algorithm here.
We use the following parameters:
$$
\gamma=\frac{\epsilon}{Ln\tau},\quad
J=\max\left(\left\lceil\frac{-\ln (\epsilon/L\tau)n}{\gamma}\right\rceil, \left\lceil\frac{2\pi n}{\gamma}\right\rceil\right).
$$
Let $\bV$
be the set of all $2L$-dimensional integer vectors of the form
$\bfv=\langle x_1,y_{1},\ldots, x_{L}, y_L \rangle$
where $1\leq x_{i}\leq J$ and $1\leq y_{i}\leq J$ for $i=1,\ldots, L$.

For each element $e\in U$, we associate it with a $2L$-dimensional integer vector
$$
\feature(e)=\langle \alpha_{1}(e),\beta_{1}(e),\ldots, \alpha_L(e), \beta_L(e) \rangle,
$$
\begin{align}
\label{eq:feature}
\text{ where }\quad
\alpha_{i}(e)= \left\lfloor  \min\left( \frac{-\ln |\Exp[\phi_{i}^{w_{e}}]|}{\gamma} , \frac{J}{n} \right)\right\rfloor
\quad\text{ and }\quad
\beta_{i}(e)=\left\lfloor \frac{\arg(\Exp[\phi_{i}^{w_{e}}])}{\gamma} \right\rfloor.
\end{align}
We call $\feature(e)$ the {\em feature vector} of $e$.
Since $|\phi_{i}|\leq 1$, we can see that $\alpha_{i}(e)\geq 0$ for any $e\in U$.
It is easy to see that $\feature(e)\in \bV$ for all $e\in U$
and $\sum_{e\in S}\feature(e)\in \bV$ for all $S\subseteq U$.
Intuitively, $\alpha_{i}(e)$ and $\beta_{i}(e)$ can be thought as the
scaled and rounded versions of $-\ln |\Exp[\phi_{i}^{w_{e}}]|$ and $\arg(\Exp[\phi_{i}^{w_{e}}])$, respectively.

We maintain $J^{2L}=(n/\epsilon)^{O(L)}$ {\em configurations} (a configuration is  just like a 
state in a dynamic program).
Each configuration $\conf(\bfv)$ is indexed by a $2L$-dimensional vector $\bfv\in \bV$
and takes 0/1 value.
In particular, the value of $\conf(\bfv)$ for each $\bfv\in \bV$ is defined as follows:
For each vector $\bfv\in \bV$,
\begin{enumerate}
\item
$\conf(\bfv)=1$ if and only if there
is a feasible solution $S\in \calF$ such that 
$\sum_{e\in S}\feature(e)=\bfv$.
\item
$\conf(\bfv)=0$ otherwise.
\end{enumerate}
For any $\bfv=\langle x_1, y_1, \ldots, x_L, y_L\rangle$, define the {\em value} of $\bfv$ to be
$$
\val(\bfv) = \sum_{k=1}^{L}c_{k} e^{-x_{k}\gamma+i y_{k}\gamma}.
$$

Lemma~\ref{lm:close} tells us the value of a configuration 
is close to the expected utility of the corresponding solution.
Lemma~\ref{lm:config} shows we can compute those configurations in polynomial time.

\begin{lemma}
\label{lm:close}
Suppose $\tmu(x)=\sum_{k=1}^{L}c_{k}\phi_k^x$,
where $|\phi_k|\leq 1$ for all $k=1,\ldots, L$.
Let $\tau=\max_k |c_k|/\umax $.
For any vector $\bfv=\langle x_1,y_{1},\ldots, x_{L}, y_L \rangle \in \bV$,
$\conf_{v}(\bfv)=1$ if and only if
there is a feasible solution $S\in \calF$ such that
$$
\Bigl|\Exp[\tmu(\len(S))]-\val(\bfv) \Bigr|= \Bigl|\Exp[\tmu(\len(S))]-\sum_{k=1}^{L}c_{k} e^{-x_{k}\gamma+i y_{k}\gamma} \Bigr| \leq O(\epsilon \umax).
$$
\end{lemma}

\begin{proof}
We first notice that
$\Exp[\tmu(\len(S))]=\Exp[\sum_{k=1}^{L} c_k \phi_k^{\len(S)}]=\sum_{k=1}^{L} c_k\Exp[ \phi_k^{\len(S)}].$
Therefore, it suffices to show that
for all $k=1,\ldots, L$,
$|\Exp[\phi_{k}^{\len(S)}]- e^{-x_k\gamma+i y_k\gamma} | \leq O(\frac{\epsilon}{L\tau}).$
Since $\conf(\bfv)=1$, we know that $\sum_{e\in S}\feature(e)=\bfv$ for some feasible solution $S\in\calF$.
In other words, we have $\sum_{e\in S}\alpha_k(e)=x_k$ and $\sum_{e\in S}\beta_k(e)=y_k$ for all $1\leq k\leq L$.

Fix an arbitrary $1\leq k\leq L$.
First, we can see that the arguments of $\Exp[\phi_{k}^{\len(S)}]$ and $e^{-x_k\gamma+i y_k\gamma}$ are close:
\begin{align*}
\left|\arg(\Exp[\phi_{k}^{\len(S)}])-y_k\gamma\right|
& \leq \sum_{e\in S}\left|\arg(\Exp[\phi_{k}^{w_{e}}])-\beta_k(e)\gamma\right|
\leq \sum_{e\in S} \gamma \leq n\gamma = \frac{\epsilon}{L\tau},
\end{align*}
where we use $\arg(c)$ to denote the argument of the complex number $c$.
Now, we show the magnitude of $\Exp[\phi_{k}^{\len(S)}]$ and $e^{-x_k\gamma+i y_k\gamma}$ are also close.
We distinguish two cases:
\begin{enumerate}
\item 
Recall that
$
\alpha_{i}(e)= \left\lfloor  \min\left( \frac{-\ln |\Exp[\phi_{i}^{w_{e}}]|}{\gamma} , \frac{J}{n} \right)\right\rfloor.
$
If there is some $e\in S$ such that $\frac{-\ln |\Exp[\phi_{i}^{w_{e}}]|}{\gamma} > \frac{J}{n} $ 
(which implies that $\alpha_k(e)=\lfloor\frac{J}{n}\rfloor$), we know that
$$
-\ln(|\Exp[\phi_{k}^{\len(S)}]|)=\sum_{e\in S}(-\ln (|\Exp[\phi_k^{w_{e}}|))>\frac{J\gamma}{n}.
$$
In this case, we have $x_{k}=\sum_{e\in S}\alpha_k(e) \geq \frac{J}{n}$. Thus, we have that
$$
\Bigl ||\Exp[\phi_{k}^{\len(S)}]|-|e^{-x_k\gamma}|\Bigr |< e^{-J\gamma/n}\leq e^{\gamma\lceil\frac{n\ln (\epsilon/L\tau)}{\gamma}\rceil/n}<\frac{\epsilon}{L\tau}.
$$
\item On the other hand, if 
$\alpha_{k}(e)= \left\lfloor \frac{-\ln |\Exp[\phi_{k}^{w_{e}}]|}{\gamma} \right\rfloor$
for all $e\in S$, 
we can see that
\begin{align*}
-\ln (|\Exp[\phi_k^{\len(S)})|)-x_k\gamma&=\sum_{e\in S}(-\ln (|\Exp[\phi^{w_{e}}|)-\alpha_k(e)\gamma)
\leq \sum_{e\in S}\gamma\leq n\gamma\leq
\frac{\epsilon}{L\tau}.
\end{align*}
Since the derivative of $e^{x}$ is less than $1$ for $x<0$, we can get that
$$
\Bigl ||\Exp[\phi_{k}^{\len(S)}]|-|e^{-x_k\gamma}|\Bigr|\leq |e^{-x_k\gamma-\epsilon/L\tau}-e^{-x_k\gamma}|\leq \frac{\epsilon}{L\tau}.
$$
\end{enumerate}

For any two complex numbers $a, b$ with $|a|\leq 1$ and $|b|\leq 1$, if $\bigl ||a|-|b|\bigr| < \epsilon$ and
$|\angle ab|=|\arg(a)-\arg(b)|< \epsilon$, we can see that
\begin{align*}
|a-b|^2 & =|a|^2+|b|^2 - 2|a| |b|\cos(\angle ab) \\
&= (|a|-|b|)^2 +2|a||b|(1-\cos(\angle ab)) \\
&\leq \epsilon^2 +2(1-\cos(\angle ab)^2) \\
& \leq \epsilon^2 +2\sin(\angle ab)^2 \\
& \leq \epsilon^2 +2|\arg(a)-\arg(b)|^2\leq 3\epsilon^2.
\end{align*}
In the third inequality, we use the fact that $\sin x<x$ for all $x>0$.
The proof is completed.\qed
\end{proof}

\begin{lemma}
\label{lm:config}
Suppose there is a pseudopolynomial time algorithm for
the exact version of $\probA$, which runs in time polynomial in $n$ and
$t$ ($t$ is the maximum integer in the instance of $\probA$).
Then, we can compute the values for all configurations $\{\conf(\bfv)\}_{\bfv\in \bV}$ in time $(\frac{n\tau}{\epsilon})^{O(L)}$.
\end{lemma}

\begin{proof}
For each vector $\bfv\in \bV$, 
we can encode it as a nonnegative integer $I(\bfv)$ upper bounded by
$J^{2L}=(\frac{n}{\epsilon})^{O(L)}$.
In particular, each coordinate of $\bfv$ takes the position of a specific digit in the integral representation,
and the base is chosen to be $J$ no carry can occur when we add at most $n$
feature vectors.
Then, determining the value of a configuration $\conf(\bfv)$ is equivalent to determining whether there is
a feasible solution $S\in \calF$ such that the total weight of $S$ (i.e., 
$\sum_{e\in S}I(\feature(e))$) is exactly the given value $I(\bfv)$.
Suppose the pseudopolynomial time algorithm for the exact version of  $\probA$ runs in time $P_{\probA}(n,t)$
for some polynomial $P_{\probA}$.
Therefore, the value of each such $\conf(\bfv)$ can be also computed in time
$P_{\probA}(n, I(\bfv))=P_{\probA}(n, (\frac{n}{\epsilon})^{O(L)})=(\frac{n}{\epsilon})^{O(L)}$.
Since $J$  are bounded by $(\frac{n\tau}{\epsilon})^{O(1)}$,
the number of configuration is $(\frac{n\tau}{\epsilon})^{O(L)}$.
The total running time is $(\frac{n\tau}{\epsilon})^{O(L)}\times (\frac{n\tau}{\epsilon})^{O(L)}=(\frac{n\tau}{\epsilon})^{O(L)}$.\qed
\end{proof}

Now, everything is ready to prove Theorem~\ref{thm:algofortmu}.

\vspace{0.3cm}
\begin{proofof} {Theorem~\ref{thm:algofortmu}}
We first use the algorithm in Lemma~\ref{lm:config} to compute the values for all configurations.
Then, we find the configuration
$\conf(\langle x_1,y_{1},\ldots, x_{L}, y_L \rangle)$ that has value $1$ and that maximizes the quantity
$|\val(\bfv)|=|\sum_{k=1}^{L}c_{k} e^{-x_{k}\gamma+i y_{k}\gamma}|$.
The feasible solution $S$ corresponding to this configuration is our final solution.
It is easy to see that the theorem follows from Lemma~\ref{lm:close}.
\qed
\end{proofof}

Theorem~\ref{thm:mainthm} can be readily obtained  from Theorem~\ref{thm:algofortmu}
and the fact $\tmu$ is an $\epsilon$-approximation of $\mu$.

\vspace{0.5cm}
\begin{proofof} {Theorem~\ref{thm:mainthm}}
Suppose $S$ is our solution and $S^{*}$ is the optimal solution for utility function $\uti$.
Recall $\tmu(x)=\sum_{k=1}^{L}c_{k}\phi_k^x$.
From Theorem~\ref{thm:algofortmu}, we know that 
$$
|\Exp[\tmu(\len(S))]| \geq |\Exp[\tmu(\len(S^{*}))]|-O(\epsilon\umax).
$$
Since $\tmu$ is an $\epsilon$-approximation of $\uti$,
we can see that
\begin{align*}
\bigl|\Exp[\uti(\len(S))]-\Exp[\tmu(\len(S))]\bigr| &= \Bigl|\int (\uti(x)-\tmu(x)) \d P_S(x) \Bigr|
 \leq \Bigl|\int \epsilon \umax \d P_S(x) \Bigr| \leq \epsilon \umax.
\end{align*}
for any solution $S$, where $P_{S}$ is the probability measure of $w(S)$.
Therefore, we have
\begin{align*}
|\Exp[\uti(\len(S))]|& \geq
 |\Exp[\tmu(\len(S))]|-\epsilon \umax \geq |\Exp[\tmu(\len(S^{*}))]|-O(\epsilon \umax) \\
& \geq |\Exp[\uti(\len(S^{*}))]|-O(\epsilon\umax). 
\end{align*}
This completes the proof of Theorem~\ref{thm:mainthm}.
\qed
\end{proofof}

\section{Class $\classA$}

The main goal of this section is to prove Theorem~\ref{thm:mainholder}.
In Section~\ref{subsec:approximateutility}, we develop a generic algorithm that
takes as a subroutine an algorithm \algo\ for approximating functions in a bounded interval domain,
and approximates $\uti(x)\in \classA$ in the infinite domain $[0,+\infty)$.
In the Section~\ref{subsec:fourier}, 
we use the Fourier series expansion as the choice of \algo\
and show that important classes of utility functions
can be approximated well.

\subsection{Approximating the Utility Function}
\label{subsec:approximateutility}

There are many works on approximating functions using short exponential sums, e.g.,
the Fourier decomposition approach \cite{stein2003fourier}, Prony's method \cite{Osborne95modified},
and many others~\cite{beylkin2002generalized,beylkin2005approximation}.
However, their approximations are done over a finite interval domain, say $[-\pi,\pi]$
or over a finite number of discrete points. No error bound can be guaranteed outside the domain.
Our algorithm is a generic procedure that
turns an algorithm that can approximate functions over $[-\pi,\pi]$ into one that can
approximate our utility function $\uti$ over $[0, +\infty)$, by utilizing the fact that $\lim_{x\rightarrow \infty} \uti(x)=0$.

Recall for $\uti\in \classA$, we assume that for any constant $\epsilon>0$, there exist a constant $T_{\epsilon}$ such that
$\uti(x)\leq \epsilon$ for $x>T_{\epsilon}$.
We also assume there is an algorithm \algo\ that, for any function $f$ (under some conditions specified later),
can produce an exponential sum
$\widehat{f}(x)=\sum_{i=1}^{L}c_{i}\phi_{i}^{x}$ which is an
$\epsilon$-approximation of $f(x)$ in $[-\pi,\pi]$
such that $|\phi_{i}|\leq 1$ and $L$ depends only on $\epsilon$ and $f$.
In fact, we can assume w.l.o.g. that \algo\ can approximate $f(x)$ over $[-B, B]$
for any $B=O(1)$.
This is because we can apply \algo\ to the scaled version
$g(x)=f(x\cdot \frac{B}{\pi})$ (which is defined on $[-\pi,\pi]$)
and then scale the obtained
approximation $\widehat{g}(x)$ back to $[-B,B]$, i.e.,
the final approximation is $\widehat{f}(x)=\widehat{g}(\frac{\pi}{B}\cdot x)$.
Scaling a function by a constant factor $\frac{B}{\pi}$ typically does not affect the smoothness of $f$ in any essential way
and we can still apply \algo.
Recall that our goal is to produce an exponential sum that is an $\epsilon$-approximation for $\uti(x)$
in $[0, +\infty)$.
We denote this procedure by \utidecomp.

\begin{center}
\fbox{
    \parbox{\linewidth}{
\begin{center}
\vspace{-0.3cm}
\textbf{Algorithm}: \utidecomp($\mu$)
\vspace{-0.3cm}
\end{center}
\begin{enumerate}
\item Initially, we slightly change function $\uti(x)$ to a new function $\barmu(x)$ as follows:
We require $\barmu(x)$ is a ``smooth '' function in $[-2T_{\epsilon}, 2T_{\epsilon}]$ such that
$\barmu(x)=\uti(x)$ for all $x\in [0,T_{\epsilon}]$;
$\barmu(x)=0$ for $|x|\geq 2T_{\epsilon}$. We choose $\barmu(x)$ in $[-2T_{\epsilon},0]$
and $[T_{\epsilon}, 2T_{\epsilon}]$ such that $\barmu(x)$ is smooth.
We do not specify the exact smoothness requirements now since
they may depend on the choice of \algo.
Note that there may be many ways to interpolate $\uti$ such that the above conditions are satisfied (see Example~\ref{ex:example1} below).
The only properties we need are: (1) $\barmu$ is amenable to algorithm \algo;
(2) $|\barmu(x)-\uti(x)|\leq \epsilon$ for $x\geq 0$.
\item We apply \algo\ to
$g(x)=\eta^{x}\barmu(x)$  over domain $[-hT_{\epsilon},hT_{\epsilon}]$
($\eta\geq 1$ and $h\geq 2$ are constants to be determined later).
Suppose the resulting exponential sum $\widehat{g}(x)=\sum_{i=1}^{L}c_{i}\phi_{i}^{x}$,
such that $|\widehat{g}(x)-g(x)|\leq \epsilon$ for all $x\in [-hT_{\epsilon},h T_{\epsilon}]$.
\item Let $\tmu(x)=\sum_{i=1}^{L}c_{i}(\frac{\phi_{i}}{\eta})^{x}$, which is our final approximation
of $\uti(x)$ on $[0,\infty)$.
\vspace{-0.3cm}
\end{enumerate}
}
}
\end{center}

\begin{example}
\label{ex:example1}
Consider the utility function
$\uti(x)=1/(x+1).$
Let $T_{\epsilon}=\frac{1}{\epsilon}-1$. So $\uti(x)<\epsilon$ for all $x>T_{\epsilon}$.
Now we create function $\barmu(x)$ according to the first step of
\utidecomp.
If we only require $\barmu(x)$ to be continuous, then we can use, for instance,
the following piecewise function:
\eat{
$$\barmu(x)= \left\{
             \begin{array}{ll}
		\frac{1}{x+1}& x\in [0,T_{\epsilon}], \\
		\frac{x}{\epsilon T}+\frac{2}{\epsilon} & x\in [T_{\epsilon},2T_{\epsilon}], \\
		0 & x>2T_{\epsilon}.  \\
		-\barmu(x) & x<0.
	\end{array}
           \right.
$$
}
$\barmu(x)=
		\frac{1}{x+1}$ for $x\in [0,T_{\epsilon}]$;
$\barmu(x)= -\frac{x}{\epsilon T_\epsilon}+\frac{2}{\epsilon}$ for $x\in [T_{\epsilon},2T_{\epsilon}]$;
$\barmu(x)= 0$ for $x>2T_{\epsilon}$;
$\barmu(x)=\barmu(-x)$ for $x<0.$
It is easy to see that $\barmu$ is continuous and $\epsilon$-approximates $\mu$. \qed
\end{example}

By setting $\eta=2$ and
\begin{align}
\label{eq:h1}
h\geq \max\left(2, \frac{\log (\sum_{i=1}^{L}|c_{i}|/\epsilon)}{T_{\epsilon}}\right),
\end{align}
we can show the following theorem.

\begin{lemma}
\label{lm:approxfunciton}
$\tmu(x)$ is a $2\epsilon$-approximation of $\uti(x)$.
\end{lemma}

\begin{proof}
We know that $|\widehat{g}(x)-g(x)|\leq \epsilon$
for $x\in [0,hT_{\epsilon}]$.
Therefore, we have that
$$
|\tmu(x)-\barmu(x)|=\left|\frac{\widehat{g}(x)}{\eta^{x}}-\frac{g(x)}{\eta^{x}}\right|\leq \frac{\epsilon}{\eta^{x}}\leq \epsilon.
$$
Combining with $|\barmu(x)-\uti(x)|\leq \epsilon$, we obtain
$|\tmu(x)-\uti(x)|\leq 2\epsilon$ for $x\in [0,hT_{\epsilon}]$.
For $x>hT_{\epsilon}$, we can see that
\begin{align*}
|\tmu(x)|&=\left|\sum_{i=1}^{L}c_{i}\left(\frac{\phi_{i}}{\eta}\right)^{x}\right|\leq \sum_{i=1}^{L}  \left|c_{i} \left(\frac{\phi_{i}}{\eta}\right)^{x}\right|
\leq \frac{1}{2^{x}}\sum_{i=1}^{L} |c_{i}| \leq \frac{1}{2^{hT_{\epsilon}}}\sum_{i=1}^{L} |c_{i}|  \leq \epsilon
\end{align*}
Since $\uti(x)<\epsilon$ for $x>hT_{\epsilon}$, the proof is complete. \qed
\end{proof}

\topic{Remark:}
Since we do not know $c_{i}$ before applying \algo, we need to set $h$ to be a quantity
(only depending on $\epsilon$ and $T_\epsilon$) such that  (\ref{eq:h1}) is always satisfied.
In particular, we need to provide an upper bound for $\sum_{i=1}^{L}|c_{i}|$.
In the next subsection, we use the Fourier series decomposition as the choice for \algo,
which allows us to provide such a bound for a large class of functions.

\subsection{Implementing {\sc Fourier}}
\label{subsec:fourier}
Now, we discuss the choice of algorithm \algo\ and the conditions that
$f(x)$ needs to satisfy so that it is possible to approximate $f(x)$
by a short exponential sum in a bounded interval.
In fact, if we know in advance that
there is a short exponential sum that can approximate $f$,
we can use the algorithms developed in
\cite{beylkin2005approximation,beylkin2010approximation} (for continuous case)
and \cite{beylkin2002generalized} (for the discrete case).
However, those works do not provide an easy characterization of the class of functions.
From now on, we restrict ourselves to the classic Fourier series technique,
which has been studied extensively and allows such characterizations.

Suppose from now on that 
$f(x)$ is a real periodic function defined on $[-\pi, \pi]$.
Consider the partial sum of the Fourier series of the function $f(x)$:
$$
(S_N f)(x)= \sum_{k=-N}^{N} c_k e^{ikx}
$$
where the Fourier coefficient $c_k = \frac{1}{2\pi} \int_{-\pi}^{\pi} f(x) e^{-ikx} \d x.$
It has $L=2N+1$ terms.
Since $f(x)$ is a real function, we have $c_k=c_{-k}$ and the partial sum is also real.
We are interested in the question under which conditions does the function $S_{N} f$
converge to $f$ (as $N$ increases) and what is convergence rate?
Roughly speaking, the ``smoother'' $f$ is, the faster $S_{N} f$ converges to $f$.
In the following, we need one classic result about the convergence of Fourier series
and show how to use it in our problem. 

We need a few more definitions.
We say $f$ satisfies the  {\em $\alpha$-H\"{o}lder condition} if
$| f(x) - f(y) | \leq C \, |x - y|^{\alpha}$, for some constant $C$ and $\alpha>0$ and any $x$ and $y$.
The constant $C$ is called the {\em H\"{o}lder coefficient} of $f$, also denoted as $|f|_{C^{0,\alpha}}$.
We say $f$ is {\em $C$-Lipschitz} if $f$ satisfies 1-H\"{o}lder condition with coefficient $C$.
\begin{example}
It is easy to check that the utility function $\uti$ in Example~\ref{ex:example1} is
1-Lipschitz since $|\frac{\d \uti(x)}{\d x}|\leq 1$ for $x\geq 0$.
We can also see that $\tchi(x)$ (defined in \eqref{ex:utility1}) is $\frac{1}{\delta}$-Lipschitz.
\end{example}
We need the following classic result of Jackson.
\begin{theorem} {\em (See e.g., \cite{powell81})}
\label{thm:fourier}
Suppose that 
$f(x)$ is a real periodic function defined on $[-\pi, \pi]$.
If $f$ satisfies the  $\alpha$-H\"{o}lder condition, it holds that
$$
|f(x)-(S_Nf)(x)|\le O\Bigl( {|f|_{C^{0,\alpha}}\ln N\over N^\alpha}\Bigr).
$$
\end{theorem}

We are ready to spell the details of \algo.
Recall $g(x)$ is obtained in step 2 in Algorithm
\utidecomp.
By construction, $g(-hT_{\epsilon})=g(hT_{\epsilon})=0$ for $h\geq 2$.
Hence, it can be considered as a periodic function with period
$2hT_{\epsilon}$.
Note that in Jackson's theorem, the periodic function $f$
is defined
on $[-\pi,\pi]$.
In order to apply Jackson's theorem to 
$g(x)$ over $[-hT_{\epsilon},hT_{\epsilon}]$,
we consider the following function $f$, which is
the scaled version of $g$:
$$
f(x)=g(xhT_{\epsilon}/\pi).
$$
Then, \algo\ returns the following function $\widehat{g}$, which is a sum of exponential functions:
$$
\widehat{g}(x) = S_N f \left(\frac{x\pi}{hT_\epsilon}\right).
$$

Now, we show that $|\widehat{g}(x)-g(x)|\leq \epsilon$ for all $x\in [-hT_{\epsilon},h T_{\epsilon}]$.
For the later parts of the analysis, we need a few simple lemmas.
The proofs of these lemmas are straightforward and thus omitted here.

\begin{lemma}
\label{lm:holderpiecewise}
Suppose $f: [a,c]\rightarrow \mathbb{R}$ is a continuous function
which consists of two pieces $f_{1}:[a,b]\rightarrow \mathbb{R}$ and $f_{2}:[b,c]\rightarrow \mathbb{R}$.
If both $f_{1}$ and $f_{2}$ satisfy the  $\alpha$-H\"{o}lder condition
with H\"{o}lder coefficient $C$, then
$|f|_{C^{0,\alpha}}\leq 2C.$
\end{lemma}

\begin{lemma}
\label{lm:holderscale}
Suppose $g: [a,c]\rightarrow \mathbb{R}$ is a continuous function
satisfying the  $\alpha$-H\"{o}lder condition with H\"{o}lder coefficient $C$.
Then, for $f(x)=g(tx)$ for some $t>0$, we have
$|f|_{C^{0,\alpha}}\leq Ct^{\alpha}.$
\end{lemma}

By Lemma~\ref{lm:holderpiecewise}, we know that the piecewise function $\widehat{\uti}$ 
(defined in step 1 in \utidecomp)
satisfies
$\alpha$-H\"{o}lder condition with coefficient $2C$.
Therefore, we can easily see that $g(x)=\barmu(x)\eta ^{x}$
satisfies $\alpha$-H\"{o}lder condition with coefficient at most $2^{1+2T_{\epsilon}}C$ on $[-hT_{\epsilon}, hT_{\epsilon}]$
(This is because $\barmu$ is non-zero only in $[-2T_{\epsilon}, 2T_{\epsilon}]$).
According to Lemma~\ref{lm:holderscale}, we have
$
|f(x)|_{C^{0,\alpha}}=|g(xhT_{\epsilon}/\pi)|_{C^{0,\alpha}}\leq 2^{1+2T_{\epsilon}}(hT_{\epsilon}/\pi)^{\alpha}C.
$
Using Theorem~\ref{thm:fourier}, 
 we obtain the following corollary.
\begin{corollary}
\label{cor:fourier}
Suppose $\uti\in \classA$ satisfies the  $\alpha$-H\"{o}lder condition
with $|\uti|_{C^{0,\alpha}}=O(1)$. 
For 
$$
N=2^{O(T_{\epsilon})}(h/\epsilon)^{1+1/\alpha},
$$
it holds that
$|g(x)-\widehat{g}(x)|\leq \epsilon$ for $x\in [-hT_{\epsilon},hT_{\epsilon}]$.
\end{corollary}
\begin{proof}
Applying Theorem~\ref{thm:fourier} to $f$
and plugging in the given value of $N$,
we can see that 
$
|f(x)-(S_Nf)(x)|\le \epsilon
$
for $x\in [-\pi,\pi]$.
Hence, we have that
$
|g(x)-\widehat{g}(x)|=
|f(\frac{x\pi}{hT_\epsilon})-S_Nf(\frac{x\pi}{hT_\epsilon})|\leq \epsilon
$
for $x\in [-hT_{\epsilon},hT_{\epsilon}]$.
\qed
\end{proof}

\vspace{0.1cm}
\topic{How to Choose $h$}:
Now, we discuss the issue left in Section~\ref{subsec:approximateutility},
that is how to choose $h$ (the value should be independent of $c_{i}$s and $L$)
to satisfy (\ref{eq:h1}),
when $\uti$ satisfies the  $\alpha$-H\"{o}lder condition for some $\alpha>1/2$.
We need the following results about the absolute convergence of
Fourier coefficients.
If $f$ satisfies the  $\alpha$-H\"{o}lder condition
for some $\alpha>1/2$,
then
$
\sum_{i=-\infty}^{+\infty} |c_{i}|\leq |f|_{C^{0,\alpha}}\cdot c_{\alpha}
$
where $c_{\alpha}$ only depends on $\alpha$ \cite{stein2003fourier}.
We can see that in order to ensure \eqref{eq:h1}, it suffices to 
 to set value $h$ such that
\begin{align*}
hT_{\epsilon} & \geq \log \frac{2^{1+2T_{\epsilon}} (hT_{\epsilon}/\pi)^{\alpha} Cc_{\alpha}}{\epsilon}
 =2T_{\epsilon}+O\bigl(\log \bigl(hT_{\epsilon}/\epsilon\bigr)\bigr).
\end{align*}
We can easily verify that the above condition can be satisfied by letting
$h=\max(O(\frac{1}{T_{\epsilon}}\log \frac{1}{\epsilon} ), 2)$.

\begin{proofof}{Theorem~\ref{thm:mainholder}}
Everything is in place to prove Theorem~\ref{thm:mainholder}.
First, we bound $L$  by Corollary~\ref{cor:fourier}: 
$$
L=2N+1=
2^{O(T_\epsilon)}\poly(1/\epsilon).
$$
Next, we bound the magnitude of each $c_k$.
Recall $c_k$ is the Fourier coefficient: 
$c_k = \frac{1}{2\pi} \int_{-\pi}^{\pi} f(x) e^{-ikx} \d x.$
where $f(x)=g(x h T_\epsilon/\pi)=\eta^{xh T_\epsilon/\pi}\barmu(x hT_\epsilon/\pi)$ for $x\in [-\pi,\pi]$.
Since $hT_\epsilon= \max(O(T_\epsilon), O(\log 1/\epsilon))$, we can see $|f(x)|\leq 2^{O(T_\epsilon)}\poly(1/\epsilon)$ for $x\in [-\pi,\pi]$.
Therefore, 
$$|c_k|\leq \frac{1}{2\pi} \int_{-\pi}^{\pi} |f(x)|\d x\leq 2^{O(T_\epsilon)}\poly(1/\epsilon).$$
Finally, combining Corollary~\ref{cor:fourier} and
Lemma~\ref{lm:approxfunciton}, we complete the proof of Theorem~\ref{thm:mainholder}.
\qed
\end{proofof}

\section{Class $\classB$}
\label{sec:concave}

In this section, we handle the case where the utility function $\uti : [0,\infty)\rightarrow [0,\infty)$ is a concave nondecreasing function
and our goal is to prove Theorem~\ref{thm:mainthmconcave}.

We use $\opt$ to denote the optimal value of our 
problem $\eum(\probA)$.
We can assume without loss of generality that we know $\opt$, modulo a multiplicative factor of $(1\pm \epsilon)$.
This can be done by guessing all powers of $(1+\epsilon)$ between 
$\max_{e\in U} \Exp[\uti(w_e)]$ and $\Exp[\uti(w(U))]$,
\footnote{
We can assume every $e\in U$ is in at least one feasible solution $S\in \calF$.
Otherwise, we can simply remove those irrelevant elements.
Then, $\opt$ is at least $\max_{e\in U} \Exp [\uti(w_e)]$.
We can test whether an item is an irrelevant element by using  
the pseudopolynomial time algorithm as follows: we assign the item with weight 1 and other items weight 0.
We ask whether there is a feasible solution with weight exactly 1.
} 
and run our algorithm for each guess.
For ease of notation, we assume that our current guess is
exact $\opt$.
Let
 \begin{align}
 \label{eq:parameter1}
 H=\opt/\epsilon^2,\quad T_1=\uti^{-1}(\opt/\epsilon) \quad \text{ and }\quad T_2=\uti^{-1}(\opt/\epsilon^2).
\end{align}
We first make the following simplifying assumption and show how to remove it later:
\begin{enumerate}
\item [S1.]
We assume $\uti(0)=0$ and $\uti(x) = \uti(T_2) =H$ for all $x>T_2$.
\end{enumerate}

\begin{figure}[t]
\begin{center}
\includegraphics[width=0.95\linewidth, height=3cm]{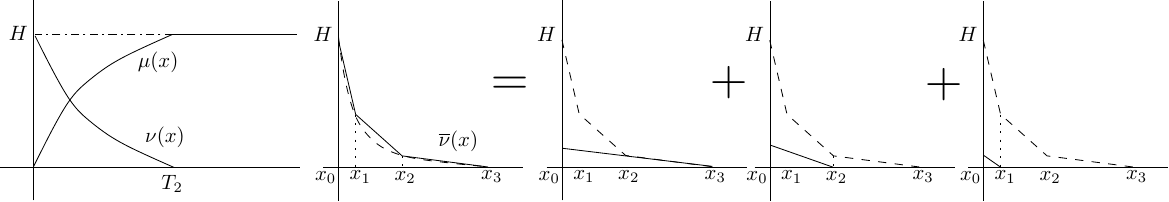}
\end{center}
\vspace{-0.5cm}
\caption{(1) The concave utility function $\uti(x)$
and $\nu(x)=H-\uti(x)$
(2) The piecewise linear function $\bnu(x)$.
(3)-(5) Decomposing $\bnu(x)$ into 
three scaled copies of $\tau(x)$.}
\vspace{-0.0cm}
\label{fig_concave}
\end{figure}

\begin{lemma}
\label{lm:concave}
If the utility function $\uti\in \classB$ satisfies the additional assumption S1,
then, for any $\epsilon>0$, we can obtain 
an exponential sum 
$\tmu(x)=\sum_{k=1}^{L} c_k \phi_k^x$,
such that $|\tmu(x)-\uti(x)|\leq O(\epsilon\opt)$ for all $x>0$,
where 
$L=\poly(1/\epsilon)$, $|c_k|\leq \poly(1/\epsilon)H$ and $|\phi_k|\leq 1$ for all $k=1,\ldots,L$.
\end{lemma}
\begin{proof}
Consider the function $\nu(x)= H-\mu(x)$.
We can see $\nu$ is a nonincreasing convex function and
$\nu(x)=0$ for all $x>T_2$.
We first approximate $\nu$ by a piecewise linear function $\bnu$ as follows.
Let $N=1/\epsilon^3$.
For all $0\leq i\leq N$,
let
 $$x_i = \mu^{-1}\left(\frac{iH}{N}\right)= \nu^{-1}\left(\frac{(N-i)H}{N}\right) \text{ and } x_{N+1}=\infty.$$
Let $h_i=\frac{\nu(x_{i+1})-\nu(x_i)}{x_{i+1}-x_i}$ for $0\leq i\leq N$.
The piecewise linear function $\bnu$ is defined by 
$\bnu(x_i)=\nu(x_i)$ for all $0\leq i\leq N$
and 
$$
\bnu(x)=\bnu(x_i) + (x-x_i)h_i, \quad\text{ for } x\in [x_i, x_{i+1}].
$$
It is easy to see $\bnu$ is also a convex function (see Figure~\ref{fig_concave})
and $|\nu(x)-\bnu(x)|\leq H/N\leq \epsilon\opt$.

Now we show $\bnu$ can be written as a linear sum of $N$ scaled copies of the 
following function $\rho$:
$$
\rho(x) =  1-x \quad\text{ for } 0\leq x\leq 1, \text{ and } \rho(x)=0 \quad\text{ for } x>1.
$$
We let 
$
\rho_{h,a}(x) = (-ha)\rho(x/a). 
$
It is easy to see that the first piece of $\rho_{h,a}$ has slope $h$ and ends at $x=a$.
Define
$$
\rho_i(x)=\rho_{h_i-h_{i+1}, x_{i+1}}(x)=x_{i+1}(h_{i+1}-h_{i})\rho\left(\frac{x}{x_{i+1}}\right)\quad \text{ for } 0\leq i\leq N.
$$
It is not hard to verify that
$
\bnu(x) = \sum_{i=0}^{N-1} \rho_{i}(x)
$
(see Figure~\ref{fig_concave}).

By Theorem~\ref{thm:mainholder}, we can find a function $\ttau(x)=\sum_{k=1}^{D}d_k\psi_k^x$
with $D=\poly(N/\epsilon^2)=\poly(1/\epsilon)$, $|d_k|=\poly(1/\epsilon)$ and $|\psi_k|\leq 1$ for $k=1,\ldots, D$,
\footnote{
It suffices to let $T_\epsilon=1$ (i.e., $\rho(x)=0$ for $x\geq 1$).
}
such that 
$
|\ttau(x)-\rho(x)|\leq \epsilon^2/N   \text{ for } x\geq 0.
$
Consider the function
$$
\tnu(x) = 
\sum_{i=0}^{N-1} x_{i+1}(h_{i+1}-h_{i})\ttau\left(\frac{x}{x_{i+1}}\right)
=\sum_{i=0}^{N-1}\sum_{k=1}^D x_{i+1} (h_{i+1}-h_i) d_k (\psi_k^{1/x_{i+1}})^x.
$$
Clearly, $\tnu$ is the summation of $ND$ exponentials.
It is not difficult to see the magnitude of each coefficient, 
$|x_{i+1} (h_{i+1}-h_i) d_k|$, is at most $-x_{i+1} h_i |d_k| \leq  \poly(1/\epsilon)H$.
We can also see that 
$$
|\tnu(x)-\nu(x)|\leq |\tnu(x)-\nu(x)|+|\tnu(x)-\bnu(x)|
\leq O(H\epsilon^2) \leq O(\epsilon\opt )  \text{ for } x\geq 0.
$$
Finally, letting $\tmu(x)=H-\tnu(x)$ finishes the proof.
\qed
\end{proof}

Since $\umax=\opt/\epsilon^2$, Lemma~\ref{lm:concave}
implies that $\tmu$ is an $\epsilon^3$-approximation of $\uti$.
Then, applying Theorem~\ref{thm:mainthm},
we can immediately obtain a polynomial time algorithm that 
runs in time $(n/\epsilon)^{\poly(1/\epsilon)}$ and finds a solution $S\in \calF$ such that
$\opt-\Exp[\uti(S)]\leq \epsilon^2 \umax \leq \epsilon \opt$, i.e., a PTAS. 

Now, we show how to get rid of the assumption S1.
From now on, the utility function $\uti$ is a general increasing concave utility function 
with $\uti(0)=0$.~\footnote{
The assumption that $\uti(0)=0$ is without loss of generality.
If $\uti(0)>0$, we can solve the problem with the new utility function $\uti(x)-\uti(0)$.
It is easy to verify a PTAS for the new problem is a PTAS for the original problem.
}
Let $\uti_H(x)=\min(\uti(x), H)$.
We can see that $\uti_H(x)$ satisfies S1.
We say a value $p\in \mathbb{R}^+$ is {\em huge} is if $p>T_2$.
Otherwise, we call it {\em normal}.
We use $\hugevalue$ to denote the set of huge values.
For each element $e$, let $w^{\vsmall}_e$ be the random variable which has
the same distribution as $w_e$ in the normal value region, and zero probability elsewhere.
For any $S\subseteq U$, let $w^{\vsmall}(S)=\sum_{e\in S}w^{\vsmall}_e$.
In the following lemma, we show $\uti_H$ is a good approximation for $\uti$
for normal values.

\begin{lemma}
\label{lm:normalapprox}
For any $S\in \calF$, we have that 
$$
\Exp[\uti(w^{\vsmall}(S))]-O(\epsilon)\opt\leq \Exp[\uti_H(w^{\vsmall}(S))] \leq \Exp[\uti(w^{\vsmall}(S))]. 
$$
\end{lemma}
\begin{proof}
It is obvious that 
$\Exp[\uti_H(w^{\vsmall}(S))] \leq \Exp[\uti(w^{\vsmall}(S))]. $
So, we only need to prove the first inequality.
For any $S\subseteq \calF$, 
we have $\Exp[\uti(w(S))]\leq \opt$.
By Markov inequality, 
$\Pr[w(S)\geq T_1] \leq \epsilon$, which implies
$\Pr[w^{\vsmall}(S)\geq T_1] \leq \epsilon$.
Now, we claim that for any integer $k\geq 1$, 
\begin{align}
\label{eq:jump}
\Pr\left[w^{\vsmall}(S)\geq (k+2)T_1\right] \leq \epsilon \Pr\left[w^{\vsmall}(S)\geq kT_1\right] .
\end{align}
Consider the following stochastic process.
Suppose the weights of the elements in $S$ are realized one by one
(say $w^{\vsmall}_{e_1},  \ldots, w^{\vsmall}_{e_n}$).
Let $Z_t$ be the sum of the first $t$ realized values.
Let $t_1$ be the first time such that 
$Z_{t_1+1}\geq kT_1$.
If this never happens, let $t_1=\infty$ and $Z_{t_1}=Z_n$.
Let event $E_1$ be $t_1\leq n$
and $E_2$ be $Z_n\geq (t+2)T_1$.
Consider the random value $Z_n-Z_{t_1}=\sum_{t=t_1+1}^n w^{\vsmall}_{e_t}$.
As $w^{\vsmall}(S)=Z_n=\sum_{t=1}^n w^{\vsmall}_{e_t}$ and all $w^{\vsmall}_{e_t}$ are nonnegative,  
we can see that 
$$
\Pr[Z_n-Z_{t_1} >  T_1 \mid E_1]
=\Pr\Bigl[\,\sum_{t=t_1+1}^n w^{\vsmall}_{e_t} >  T_1 \mid t_1\leq n\,\Bigr]
\leq \Pr\Bigl[\,\sum_{t=1}^n w^{\vsmall}_{e_t} >  T_1\,\Bigr]
\leq \epsilon. $$
Moreover, we can see that 
event $E_1\wedge (Z_n-Z_{t_1} >  T_1)$ is a necessary condition for event $E_2$.
Hence, the claim holds because 
$$
\hspace{4cm}
\Pr[E_2]\leq \Pr[E_1] \Pr[Z_n-Z_{t_1} >  T_1 \mid E_1] \leq \epsilon \Pr[E_1]. 
\hspace{4cm}
$$
From \eqref{eq:jump}, we can see that 
$
\Pr\left[w^{\vsmall}(S)\geq 3T_1\right]\leq \epsilon^2,
$
$
\Pr\left[w^{\vsmall}(S)\geq 5T_1\right]\leq \epsilon^3, \ldots
$
so on and so forth.
Furthermore, we can see that
\begin{align*}
\Exp[\uti(w^{\vsmall}(S))]- \Exp[\uti_H(w^{\vsmall}(S))] 
&=\int_{H}^{\infty} \Pr[\uti(w^{\vsmall}(S))\geq x] \d x 
=\int_{H}^{\infty} \Pr[w^{\vsmall}(S)\geq \uti^{-1}(x)] \d x \\
&=\frac{\opt}{\epsilon}\int_{0}^{\infty} \Pr[w^{\vsmall}(S)\geq \uti^{-1}(H+k\opt/\epsilon)] \d k \\
&\leq \frac{\opt}{\epsilon}\int_{0}^{\infty} \Pr[w^{\vsmall}(S)\geq T_2+kT_1] \d k \\
&\leq \frac{2\opt}{\epsilon}\sum_{k=2/\epsilon}^{\infty} \epsilon^{k} 
\leq O(\epsilon \opt). 
\end{align*}
The first inequality holds due to the concavity of $\uti$
(or equivalently, the convexity of $\uti^{-1}$):
$$
\hspace{1cm}
\uti^{-1}(H+k\opt/\epsilon) \geq T_2+\uti^{-1}(k\opt/\epsilon) =
T_2+k \uti^{-1}(\opt/\epsilon) = T_2+k T_1 \text{ for }k \geq 0.
\hspace{1cm}
\Box
$$
\eat{
Consider a realization $\realization$ in which $w^{\vsmall}(S)=p$ for some $p\in [(t+2)T_1, (t+3)T_1] $.
Since the support of each $w^{\vsmall}_e$ is at most $T_1$,
we know that in realization $\realization$
there exists a subset $S'_{\realization}\subset S$
such that 
$w^{\vsmall}(S')=p'$ for some $p\in [tT_1, (t+1)T_1] $
(if there are many of them, we fix an arbitrary one).
Let $S''_{\realization}=S\setminus S'_{\realization}$.
\begin{align}
\Pr\left[w^{\vsmall}(S)\in [(t+2)T_1, (t+3)T_1] \right] 
= \sum_{\realization} \Pr[\realization] E_1[\realization]
= \sum_{\realization} \Pr_{S'_\realization}[\realization(S'_\realization)] 
\Pr_{S''_\realization}[\realization(S''_\realization)] 
E_1[\realization] 
\end{align}
}

\end{proof}

Now, we handle the contribution from huge values.
Let $\hugevalue=\{p\mid p \geq T_2\}$ and
$$
\Hg(e)=\sum_{p\in \hugevalue} \Pr[w_e=p] \uti(p).
\footnote{
If $w_e$ is continuously distributed, we let
$\Hg(e)=\int_{T_2}^\infty \uti(x) p_e(x)\d x$.
The algorithm and analysis for the continuous case are exactly the same as the discrete case.
}
$$ 
$\Hg(e)$ can be thought as the expected contribution of huge values of $e$.
We need the following observation in \cite{bhalgat2014utility}:
the contribution of the huge values can be essentially linearized and separated from the contribution
of normal values, in the sense of the following lemma.
We note that the simple insight has been used in a variety of contexts
in stochastic optimization problems (e.g., \cite{munteanu2014smallest,huang2015approximating,li2014epsilon}).

\begin{lemma} 
\label{lm:hugevalue}
{\em (The first half of Theorem 2 in \cite{bhalgat2014utility})}
For any $S\in \calF$, we have that 
$$
\Exp[\uti(w(S))] \in (1\pm O(\epsilon)) \Bigl(\Exp[\uti(w^\vsmall (S))] + \sum_{e\in S} \Hg(e)\Bigr). 
$$
\end{lemma}

Now, we are ready to state our algorithm, which is an extension of the algorithm in 
Section~\ref{sec:algorithm}.
Using Lemma~\ref{lm:concave}, we first obtain a function $\tmu_H(x) = \sum_{k=1}^L c_k \phi^x_k$
such that $|\tmu_H(x)-\uti_H(x)|\leq \epsilon\opt$. 
The feature vector $\feature(e)$ is a $2L+1$-dimensional integer vector
$$
\feature(e)=\langle \alpha_{1}(e),\beta_{1}(e),\ldots, \alpha_L(e), \beta_L(e), \lfloor n\Hg(e)/\epsilon\opt\rfloor \rangle,
$$
where $\alpha_i(e), \beta_i(e)$ are defined as in \eqref{eq:feature} with respect to $w^{\vsmall}_e$.
In other words, we extend the original feature vector by one more coordinate which represents
the (scaled and rouned) contribution of huge values.
Similarly, each configuration $\conf(\bfv)$ is indexed by such a $2L+1$-dimensional vector $\bfv$.
The last coordinate of $\bfv$ is at most $n^2/\epsilon$.
As before, we let $\conf(\bfv)=1$ if and only if there
is a feasible solution $S\in \calF$ such that 
$\sum_{e\in S}\feature(e)=\bfv$.
We slightly modify the definition of $\val(\bfv)$ to incorporate the contribution of huge values, as the following:
$$
\val(\langle x_{1}, y_1,\ldots,x_L, y_L,  z \rangle)|=
\sum_{k=1}^{L}c_{k} e^{-x_{k}\gamma+i y_{k}\delta}+ z\cdot \frac{\epsilon \opt}{n}.
$$
Using the same technique as in Lemma~\ref{lm:config} and the pseudopolynomail time algorithm for 
$\probA$, we can compute the values of all configurations in time $(n/\epsilon)^{\poly(1/\epsilon)}$.
Then, we return the solution for which the corresponding configuration 
$\conf(\langle x_{1}, y_1,\ldots,x_L, y_L,  z \rangle)$ 
that takes value 1 and
maximizes $|\val(\bfv)|$.

\vspace{0.5cm}
\begin{proofof} {Theorem~\ref{thm:mainthmconcave}}
The proof is similar to that of Theorem~\ref{thm:mainthm}.
Let any $S\subseteq U$, let 
$\bfv_S=\sum_{e\in S}\feature(e)$
Using the same proof of Lemma~\ref{lm:close} 
and the fact that $|\tmu_H(x)-\uti_H(x)|\leq \epsilon\opt$, 
we can see that for any $S\in \calF$,
$$
|\val(\bfv_S)|=\Exp[\mu(w^{\vsmall}(S))] + \sum_{e\in S}\Hg(e)\pm O(\epsilon \opt).
$$
Combining with 
Lemma~\ref{lm:normalapprox} and 
Lemma~\ref{lm:hugevalue},
we can further see that for any $S\in \calF$,
$$
|\val(\bfv_S)|=(1\pm O(\epsilon))\Exp[\mu(w(S))]\pm O(\epsilon \opt).
$$
Suppose $S$ is our solution and $S^{*}$ is the optimal solution for utility function $\uti$.
From our algorithm, we know that 
$|\val(\bfv_S)|\geq |\val(\bfv_{S^*})|$, which implies
$\Exp[\mu(w(S))]\geq(1-O(\epsilon))\opt$
and completes the proof.
\qed
\end{proofof}

\section{Class $\classC$}
\label{sec:classC}

Recall that $\uti(x)\in\classC$ is a positive, differentiable and increasing function
and $\frac{\d}{\d x}\uti(x)\in [\lb, \ub]$ for some constants $\lb,\ub>0$ and all $x\geq 0$.
By scaling, we can assume without loss of generality that 
$\lb\leq 1\leq \ub$.
Our algorithm is almost the same as the one in Section~\ref{sec:concave} except that 
we use a slightly different set of parameters:
$$
H=\frac{\opt}{\epsilon^2}\cdot \frac{\lb}{\ub},\quad
T_1=\frac{\opt}{\epsilon}\cdot \frac{\lb}{\ub}, \quad
T_2=\uti^{-1}\left(\frac{\opt}{\epsilon^2}\cdot \frac{\lb}{\ub}\right),\quad
\text{ and }\quad\hugevalue=\{p\mid p \geq T_2\}
$$

Let $\uti_H(x)=\min(\uti(x), H)$.
So, $\uti_H$ satisfies assumption S1.
However, we can not use Lemma~\ref{lm:concave} since it requires concavity.
Nevertheless, we can still approximate $\uti_H$ by a short exponential sum, as in Lemma~\ref{lm:increasing}.
The remaining algorithm is exactly the same as the one in Section~\ref{sec:concave}.
To prove the performance guarantee, we only need to prove analogues of
Lemma~\ref{lm:normalapprox}
and 
Lemma~\ref{lm:hugevalue}. 
Now, we prove the aforementioned lemmas.

\begin{lemma}
\label{lm:increasing}
For any $\epsilon>0$, we can obtain 
an exponential sum 
$\tmu_H(x)=\sum_{k=1}^{L} c_k \phi_k^x$,
such that $|\tmu_H(x)-\uti_H(x)|\leq O(\epsilon\opt)$ for all $x>0$,
where 
$L=\poly(1/\epsilon)$, $|c_k|\leq \poly(1/\epsilon)H$ and $|\phi_k|\leq 1$ for all $k$.
\end{lemma}
\begin{proof}
Since $\frac{\d}{\d x}\uti(x)\in [\lb, \ub]$, we can see that 
$H/\ub\leq T_2\leq H/\lb$.
Consider $\nu(x)= H-\uti_H(x)$.
We can see $\nu$ is a decreasing, differentiable function and
$\nu(x)=0$ for all $x>T_2$.
Consider the function $\bnu(x)=\frac{1}{H}\nu(xH)$.
First, let $T_{\epsilon}= 1/\lb \geq T_2/H$ and we can see
$\bnu(x)=0$ for $x>T_{\epsilon}$.
Hence,  $\bnu\in \classA$ and satisfies 
$\ub$-Lipschitz condition.
By Theorem~\ref{thm:mainholder}, we can compute 
a function $\tnu(x)=\sum_{k=1}^L d_k \psi^x_k$, which is 
an $\epsilon^2$-approximation of $\nu$, with $L=\poly(1/\epsilon)$, $|\psi_k|\leq 1$ 
and $d_k=\poly(1/\epsilon)$
for all $k$.
Therefore, $\tmu_H(x)=H-H\tnu(x/H)=H-\sum_{k=1}^L (Hd_k) (\psi_k^{1/H})^x$
is the desired approximation.
\qed
\end{proof}

The following lemma is an analogue of Lemma~\ref{lm:normalapprox}.
\begin{lemma}
\label{lm:normalapprox2}
For any $S\in \calF$, we have that 
$$
\Exp[\uti_H(w^{\vsmall}(S))] \in (1\pm O(\epsilon)) \Exp[\uti(w^{\vsmall}(S))]. 
$$
\end{lemma}
\begin{proof}
The proof is almost the same as that of Lemma~\ref{lm:normalapprox}, 
except that the last line makes use of the bounded derivative assumption
(instead of the concavity):
$$
\hspace{3.5cm}
\uti^{-1}\Bigl(H+\frac{k\opt}{\epsilon}\Bigr) \geq T_2+\frac{k\opt}{\epsilon \ub} \geq T_2+k T_1 \text{ for }k > 0.
\hspace{3.5cm}
\Box
$$
\eat{
Consider a realization $\realization$ in which $w^{\vsmall}(S)=p$ for some $p\in [(t+2)T_1, (t+3)T_1] $.
Since the support of each $w^{\vsmall}_e$ is at most $T_1$,
we know that in realization $\realization$
there exists a subset $S'_{\realization}\subset S$
such that 
$w^{\vsmall}(S')=p'$ for some $p\in [tT_1, (t+1)T_1] $
(if there are many of them, we fix an arbitrary one).
Let $S''_{\realization}=S\setminus S'_{\realization}$.
\begin{align}
\Pr\left[w^{\vsmall}(S)\in [(t+2)T_1, (t+3)T_1] \right] 
= \sum_{\realization} \Pr[\realization] E_1[\realization]
= \sum_{\realization} \Pr_{S'_\realization}[\realization(S'_\realization)] 
\Pr_{S''_\realization}[\realization(S''_\realization)] 
E_1[\realization] 
\end{align}
}

\end{proof}

We handle the contribution from huge values in the same way.
Recall
$
\Hg(e)=\sum_{p\in \hugevalue} \Pr[w_e=p] \uti(p).
$ 
The following lemma is an analogue of Lemma~\ref{lm:hugevalue}.

\begin{lemma} 
\label{lm:hugevalue2}
For any $S\in \calF$, we have that 
$$
\Exp[\uti(w(S))] \in (1\pm O(\epsilon)) \Bigl(\Exp[\uti(w^\vsmall (S))] + \sum_{e\in S} \Hg(e)\Bigr). 
$$
\end{lemma}
\begin{proof}
We can use exactly the same proof of Theorem 2 in \cite{bhalgat2014utility} to show that
$\Exp[\uti(w(S))] \geq (1- O(\epsilon)) (\Exp[\uti(w^\vsmall (S))] + \sum_{e\in S} \Hg(e))$, 
as the proof holds even without the concavity assumption.
The other direction requires a different argument, which goes as follows.
Let $E_0$ be the event that no $w_e$ is realized to a huge value
and $E_{e,p}$ be the event that $w_e$ is realized to value $p\in \hugevalue$.
By Markov inequality, we have $\Pr[E_0]\geq 1-\epsilon\lb/\ub$.
Moreover, using the fact that $e^{-x}\geq 1-x$, we have that
\eat{
\begin{align*}
\sum_{e\in S, p\in \hugevalue} \Pr[E_{e,p}] 
& \leq 
\Pr \bigl[\bigcup_{e\in S, p\in \hugevalue}  E_{e,p}\, \bigr] 
+\sum_{e\in S, p\in \hugevalue}\sum_{e'\ne e, p'\in \hugevalue}
\Pr[E_{e,p}\wedge E_{e',p'}] \\
& \leq
\epsilon + \Biggl( \sum_{e\in S, p\in \hugevalue} \Pr[E_{e,p}] \Biggr)^2
\end{align*}
}
\begin{align*}
\exp\Bigl(-\sum_{e\in S, p\in \hugevalue} \Pr[E_{e,p}] \Bigr)
&=
\prod_{e\in S}  \exp\Bigl( -\sum_{p\in \hugevalue}\Pr[E_{e,p}] \Bigr) \\
& \geq
\prod_{e\in S} \Bigl(1- \sum_{p\in \hugevalue}\Pr[E_{e,p}] \Bigr) \\
&=\Pr[E_0]\geq 1-\epsilon\lb/\ub\geq 
1-\epsilon.
\end{align*}
Hence, 
$\sum_{e\in S, p\in \hugevalue} \Pr[E_{e,p}]\leq -(\ln(1-\epsilon))\leq 2\epsilon$
for $\epsilon<1/2$.

Next, we can see that
$
\Exp[\uti(w(S)) \mid E_0] \Pr[E_0] \leq \Exp[\uti(w^{\vsmall}(S)) ]
$
(for each realization of $\{w_e\}_{e\in S}$ satisfying $E_0$, there is 
a corresponding realization of $\{w^{\vsmall}_e\}_{e\in S}$).
From the bounded derivative assumption, we can also see that
$
\Exp[\uti(w(S)) \mid E_{e,p}]\leq \mu(p) + \ub \cdot\Exp[w(S)]
$
By inclusion-exclusion, we have that
\begin{align*}
\Exp[\uti(w(S))] 
& \leq \Exp[\uti(w(S)) \mid E_0] \Pr[E_0] 
+ \sum_{e\in S, p\in \hugevalue}\Exp[\uti(w(S)) \mid E_{e,p}] \Pr[E_{e,p}] \\
& \leq \Exp[\uti(w^{\vsmall}(S)) ] 
+ \sum_{e\in S, p\in \hugevalue} \Pr[E_{e,p}] \mu(p)
+ \sum_{e\in S, p\in \hugevalue}\Pr[E_{e,p}] \cdot \ub \cdot \Exp[w(S)]\\
& \leq \Exp[\uti(w^{\vsmall}(S)) ] 
+ \sum_{e\in S}\Hg(e)
+ O(\epsilon) \Exp[\uti(w(S))].
\end{align*}
The last inequality holds since 
$\Exp[w(S)]\leq \Exp[\mu(w(S))]/\lb = O(\Exp[\mu(w(S))])$. 
\qed
\end{proof}

\section{Applications}
\label{sec:app}

We first consider two utility functions $\chi(x)$ and $\tchi(x)$
presented in the introduction.
Note that maximizing $\Exp[\chi(w(S))]$ is equivalent to maximizing $\Prob(w(S)\leq 1)$.
The following lemma is straightforward.

\begin{lemma}
\label{lm:sanwich}
For any solution $S$,
$$
\Prob(w(S)\leq 1)\leq \Exp[\tchi(w(S))] \leq \Prob(w(S)\leq 1+\delta).
$$
\end{lemma}

\begin{corollary}
\label{cor:thres}
Suppose there is a pseudopolynomial time algorithm for
the exact version of $\probA$.
Then, for any fixed constants $\epsilon>0$ and $\delta>0$,
there is an algorithm
that runs in time $(\frac{n}{\epsilon})^{\poly(1/\epsilon)}$,
and produces a solution $S\in \calF$
such that
$$
\Prob(w(S)\leq 1+\delta) +\epsilon \geq \max_{S'\in \calF}\Prob(w(S')\leq 1)
$$
\end{corollary}
\begin{proof}
By Theorem~\ref{thm:mainthm}, Theorem~\ref{thm:mainholder} and Lemma~\ref{lm:sanwich}, we can easily
obtain the corollary.
Note that we can choose $T_{\epsilon}=2$ for any  $\delta\in (0,1)$ and $\epsilon>0$. Thus 
$L=\poly(1/\epsilon)$. \qed
\end{proof}

Now, let us see some applications of our general results
to specific problems.

\topic{Stochastic Shortest Path}:
Finding a path with the exact target length (we allow non-simple paths)\footnote{
The exact version of {\em simple} path is NP-hard, since it
includes the Hamiltonian path problem as a special case.
}
 can be easily done in pseudopolynomial time
by dynamic programming.

\topic{Stochastic Spanning Tree}:
We are given a graph $G$, where the weight of each edge $e$
is an independent, nonnegative random variable. 
Our objective is to find a spanning tree $T$ in $G$,
such that $\Prob(w(T)\leq 1)$ is maximized.
Polynomial time algorithms have been developed for Gaussian distributed edges~\cite{ishii1981stochastic,geetha1993stochastic}.
To the best of our knowledge, no approximation algorithm with provable guarantee is known for other distributions.
Noticing there exists a pseudopolynomial time algorithm for the exact spanning tree problem \cite{barahona1987exact},
we can directly apply Corollary~\ref{cor:thres}.

\topic{Stochastic $k$-Median on Trees}:
The problem asks for a set $S$ of $k$ nodes in the given probabilistic tree $G$
such that $\Prob(\sum_{v\in V(G)} \dist(v,S)\leq 1)$ is maximized,
where $\dist(v,S)$ is the minimum distance from $v$ to any node in $S$ in the tree metric.
The $k$-median problem can be solved optimally in polynomial time on trees
by dynamic programming \cite{kariv1979algorithmic}.
It is straightforward to modify the dynamic program to get a pseudopolynomial time algorithm
for the exact version.

\topic{Stochastic Knapsack with Random Sizes}:
We are given a set $U$ of $n$ items.
Each item $i$ has a random size
$w_{i}$ and a deterministic profit $v_{i}$.
We are also given a positive constant $0\leq \gamma\leq 1$.
The goal is to find a subset $S\subseteq U$ such that
$\Prob(w(S)\leq 1)\geq \gamma$
and the total profit $v(S)=\sum_{i\in S}v_{i}$ is maximized.

If the profits of the items are polynomially bounded integers,
we can see the optimal profit is also a polynomially bounded integer.
We can first guess the optimal profit.
For each guess $g$, we solve the following problem:
find a subset $S$ of items such that the total profit of $S$ is exactly $g$
and $\Exp[\tchi(w(S))]$ is maximized.
The exact version of the deterministic problem is to find a solution $S$
with a given total size and a given total profit, which can be easily solved
in pseudopolynomial time by dynamic programming.
Therefore, by Corollary~\ref{cor:thres}, we can easily show
that we can find in polynomial time a set $S$ of items such that the total profit $v(S)$ is at least the optimum
and $\Prob(w(S)\leq 1+\epsilon)\geq (1-\epsilon) \gamma$ for any constant $\epsilon$ and $\gamma$.

If the profits are general integers, we can use the standard scaling
technique to get a $(1-\epsilon)$-approximation for the total profit.
We first make a guess of the optimal profit, rounded down to the nearest power of
$(1+\epsilon)$.
There are at most $\log_{1+\epsilon} \frac{n\max_{i}v_{i}}{\min_{i}v_{i}}$ guesses.
For each guess $g$, we solve the following problem.
We discard all items with a profit larger than $g$.
Let $\Delta=\frac{\epsilon g}{n^{2}}$.
For each item with a profit smaller than $\frac{\epsilon g}{n}$, we set its new profit to be $\bar{v}_{i}=0$.
Then, we scale each of the rest profits $v_{i}$ to $\bar{v}_{i}=\Delta\lfloor \frac{v_{i}}{\Delta}\rfloor$.
Now, we define the feasible set
$$\calF(g)=\{S \mid \sum_{i\in S} (1-2\epsilon)g\leq \sum_{i\in S}\bar{v}_{i} \leq (1+2\epsilon)g\}.$$
Since there are at most $\frac{n^{2}}{\epsilon}$ distinct $\bar{v}$ values,
we can easily show that finding a solution $S$ in $\calF(g)$
with a given total {\em size} can be solved in pseudopolynomial time by dynamic programming.

Denote the optimal solution by $S^{*}$ and the optimal profit by $OPT$.
Suppose $g$ is the right guess, i.e., $(\frac{1}{1+\epsilon})OPT\leq g\leq OPT$.
We can easily see that for any solution $S$, we have that
$$
(1-\frac{1}{n})\sum_{i\in S}v_{i} -\epsilon g\leq
\sum_{i\in S}\bar{v}_{i}\leq
\sum_{i\in S}v_{i}
$$
where the first inequalities are due to $v_{i}\geq \frac{\epsilon g}{n}$ and we set at most $\epsilon g$
profit to zero.
Therefore, we can see $S^{*}\in \calF(g)$.
Applying Corollary~\ref{cor:thres},
we obtain a solution $S$ such that
$
\Prob(w(S)\leq 1+\delta) +\epsilon \geq \Prob(w(S^{*})\leq 1+\delta).
$
Moreover, the profit of this solution
$
v(S)=\sum_{i\in S}v_{i} \geq \sum_{i\in S}\bar{v}_{i} \geq (1-2\epsilon) g \geq (1-O(\epsilon))OPT.
$

In sum, we have obtained the following result.
\begin{theorem}
\label{cor:knapsack}
For any constants $\epsilon>0$  and $\gamma>0$, there is a polynomial time algorithm
to compute a set $S$ of items such that the total profit $v(S)$ is within a $1-\epsilon$ factor of the optimum
and $\Prob(w(S)\leq 1+\epsilon)\geq (1-\epsilon) \gamma$.
\end{theorem}
Bhalgat et al.~\cite[Theorem 8.1]{bhalgat10} obtained the same result, with a running time $n^{2^{\poly(1/\epsilon)}}$,
while our running time is $n^{\poly(1/\epsilon)}$.

Moreover, we can easily extend our algorithm to generalizations of the knapsack problem
if the corresponding exact version has a pseudopolynomial time algorithm.
For example, we can get the same result for the partial-ordered knapsack problem with tree constraints \cite{garey-johnson:79,safer04fully}.
In this problem, items must be chosen in accordance with specified precedence constraints and these precedence constraints form
a partial order and the underlining undirected graph is a tree (or forest). A pseudopolynomial algorithm for
this problem is presented in \cite{safer04fully}.

\topic{Stochastic Knapsack with Random Profits}:
We are given a set $U$ of $n$ items.
Each item $i$ has a deterministic size
$w_{i}$ and a random profit $v_{i}$.
The goal is to find a subset of items that can be packed into a knapsack with capacity $1$
and the probability that the profit is at least a given threshold $T$ is maximized.
Henig \cite{henig1990risk} and Carraway et al. \cite{carraway1993algorithm} studied this problem for normally distributed profits
and presented dynamic programming and branch and bound heuristics to solve this problem optimally.

We can solve the equivalent problem of minimizing the probability that the profit is at most the given threshold,
subject to the capacity constraint.
We first show that relaxing the capacity constraint is necessary.
Consider the following deterministic knapsack instance. 
The profit of each item is the same as its size.
The given threshold is $1$. We can see that the optimal probability is $1$ if and only if there is a subset of items of total size
exactly $1$. Otherwise, the optimal probability is $0$.
However, determining whether these is a subset of items 
with total size exactly 1 is NP-Complete.
Therefore, it is NP-hard to approximate the original problem within any additive error less than $1$
without violating the capacity constraint.

The corresponding exact version of the deterministic problem is to find
a set of items $S$ such that $w(S)\leq 1$ and $v(S)$ is equal to a given target value.
In fact, there is no pseudopolynomial time algorithm for this problem.
Since otherwise we can get an $\epsilon$ additive approximation without violating the capacity constraint,
contradicting the lower bound argument.
Note that a pseudopolynomial time algorithm here should run in time polynomial in
the profit value (not the size).
However, if the sizes can be encoded in $O(\log n)$ bits (we only have a polynomial
number of different sizes), we can solve the problem
in time polynomial in $n$ and the largest profit value by standard dynamic programming.

For general sizes, we can round the size of each item down to the nearest multiple of $\frac{\delta}{n}$.
Then, we can solve the exact version in pseudopolynomial time $\poly(\max_i v_i, n, 1/\delta)$ by dynamic programming.
It is easy to show that for any subset of items, its total size is at most the total rounded size plus $\delta$.
Therefore, the total size of our solution is at most $1+\delta$.
We summarize the above discussion in the following theorem.

\begin{theorem}
If the optimal probability is $\Omega(1)$, we can find in time $(n/\epsilon\delta)^{\poly(1/\epsilon)}$
a subset $S$ of items such that
$\Prob(v(S)>(1-\epsilon)T)\geq (1-\epsilon)\opt$ and $w(S)\leq 1+\delta$, for any constant $\epsilon>0$.
\end{theorem}


\section{Extensions}
\label{sec:extension}
In this section, we discuss some extensions to our basic approximation scheme.
We first consider optimizing a constant number of utility functions in Section~\ref{subsec:multiutility}.
Then, we study the problem where the weight of each element is a random vector in Section~\ref{subsec:multiweight}.

\subsection{Multiple Utility Functions}
\label{subsec:multiutility}

The problem we study in this section contains a set $U$ of $n$ elements.
Each element $e$ has a random weight $w_{e}$.
We are also given $d$ utility functions $\mu_{1},\ldots, \mu_{d}$
and $d$ positive numbers $\lambda_{1},\ldots,\lambda_{d}$.
We assume $d$ is a constant.
A feasible solution consists of $d$ subsets of elements that satisfy some property.
Our objective is to find a feasible solution $S_{1},\ldots, S_{d}$ such that
$\Exp[\mu_{i}(w(S_{i}))]\geq \lambda_{i}$ for all $1\leq i\leq d$.

We can easily extend our basic approximation scheme
to the multiple utility functions case as follows.
We decompose these utility functions into short exponential sums using \utidecomp\ as before.
Then, for each utility function, we maintain $(n/\epsilon)^{O(L)}$ configurations.
Therefore, we have $(n/\epsilon)^{O(dL)}$ configurations in total and we would like to compute
the values for these configurations.
We denote the deterministic version of the problem under consideration by $\probA$.
The exact version of $\probA$ asks for a feasible solution $S_{1},\ldots, S_{d}$
such that the total weight of $S_{i}$ is exactly the given number $t_{i}$ for all $i$.
Following an argument similar to Lemma~\ref{lm:config}, we can easily get the following
generalization of Theorem~\ref{thm:mainthm}.

\begin{theorem}
\label{thm:multiutility}
Assume that there is a pseudopolynomial algorithm for the exact version of $\probA$.
Further assume that given any $\epsilon > 0$,
we can $\epsilon$-approximate each utility function by an exponential sum with at most $L$ terms.
Then, there is an algorithm that runs in time $(n/\epsilon)^{O(dL)}$ and
finds a feasible solution $S_{1},\ldots, S_{d}$ such that $\Exp[\mu_{i}(w(S_{i})]\geq \lambda_{i}-\epsilon$
for $1\leq i\leq d$, if there is a feasible solution for the original problem.
\end{theorem}

Now let us consider two simple applications of the above theorem.

\topic{Stochastic Multiple Knapsack}:
In this problem we are given a set $U$ of $n$ items, $d$ knapsacks with capacity $1$,
and $d$ constants $0\leq \gamma_{i}\leq 1$. We assume $d$ is a constant.
Each item $i$ has a random size $w_{i}$ and a deterministic profit $v_{i}$.
Our objective is to find $d$ disjoint subsets $S_{1},\ldots, S_{d}$ such that
$\Prob(w(S_{i})\leq 1) \geq \gamma_{i}$ for all $1\leq i\leq d$ and $\sum_{i=1}^{d}v(S_{i})$ is maximized.
The exact version of the problem is to
find a packing such that the load of each knapsack $i$ is {\em exactly} the given value $t_{i}$.
It is not hard to show this problem can be solved
in pseudopolynomial time by standard dynamic programming.
If the profits are general integers, we also need the scaling technique
as in stochastic knapsack  with random sizes. In sum, we can get the following
generalization of Theorem~\ref{cor:knapsack}.

\begin{theorem}
\label{thm:multiknapsack}
For any constants $d\in \mathbb{N}$, $\epsilon>0$ and $0\leq \gamma_{i}\leq 1$ for $1\leq i\leq d$,
there is a polynomial time algorithm to compute
$d$ disjoint subsets $S_{1},\ldots, S_{d}$ such that the total profit $\sum_{i=1}^{d} v(S_{i})$ is within a $1-\epsilon$ factor of the optimum
and $\Prob(w(S_{i})\leq 1+\epsilon)\geq (1-\epsilon) \gamma_{i}$ for $1\leq i\leq d$.
\end{theorem}

\topic{Stochastic Multidimensional Knapsack}:
In this problem we are given a set $U$ of $n$ items and a constant $0\leq \gamma\leq 1$.
Each item $i$ has a deterministic profit $v_{i}$ and a random size which is a random $d$-dimensional vector
$\bw_{i}=\{w_{i1},\ldots, w_{id}\}$. We assume $d$ is a constant.
Our objective is to find a subset $S$ of items such that
$\Prob(\bigwedge_{j=1}^{d} (\sum_{i\in S} w_{ij}\leq 1)) \geq \gamma$ and total profit is maximized.
This problem can be also thought as the fixed set version of the
stochastic packing problem considered in \cite{dean2005adaptivity,bhalgat10}.
We first assume the components of each size vector are independent.
The correlated case will be addressed in the next subsection.

For ease of presentation, we assume $d=2$ from now on. Extension to general constant $d$ is straightforward.
We can solve the problem by casting it into a multiple utility problem as follows.
For each item $i$, we create two copies $i_{1}$ and $i_{2}$.
The copy $i_{j}$ has a random weight $w_{ij}$.
A feasible solution consists of two sets $S_{1}$ and $S_{2}$ such that
$S_{1}$ ($S_{2}$) only contains the first (second) copies of the elements and
$S_{1}$ and $S_{2}$ correspond to exactly the same subset of original elements.
We enumerate all such pairs $(\gamma_{1}, \gamma_{2})$ such that
$\gamma_{1}\gamma_{2}\geq \gamma$ and $\gamma_{i} \in [\gamma, 1]$
is a power of $1-\epsilon$ for $i=1,2$. Clearly, there are a polynomial number of
such pairs. For each pair $(\gamma_{1},\gamma_{2})$,
we solve the following problem:
find a feasible solution $S_{1}, S_{2}$ such that
$\Prob(\sum_{i\in S_{j}} w_{ij}\leq 1) \geq \gamma_{j}$ for all $j=1,2$ and total profit is maximized.
Using the scaling technique and
Theorem~\ref{thm:multiutility} for optimizing multiple utility functions, we can get a $(1-\epsilon)$-approximation
for the optimal profit and $\Prob(\bigwedge_{j=1}^{2} (\sum_{i\in S_{j}} w_{ij}\leq 1))=
\prod_{j=1}^{2}\Prob(\sum_{i\in S_{j}} w_{ij}\leq 1)
\geq (1-O(\epsilon))\gamma_{1}\gamma_{2}\geq (1-O(\epsilon))\gamma$.

We note that the same result for independent components
can be also obtained by using the discretization technique developed
for the adaptive version of the problem in \cite{bhalgat10}
\footnote{With some changes of the discretization technique, the correlated case can be also handled \cite{bhalgat10note}.}.
If the components of each size vector are correlated,
we can not decompose the problem into two $1$-dimensional utilities
as in the independent case.
Now, we introduce a new technique to handle the correlated case.

\subsection{Multidimensional Weight}
\label{subsec:multiweight}
The general problem we study contains a set $U$ of $n$ elements.
Each element $e$ has a random weight vector
$w_{i}=(w_{i1},\ldots, w_{id})$.
We assume $d$ is a constant.
We are also given a utility function $\mu: \mathbb{R}^{d}\rightarrow \mathbb{R}^{+}$.
A feasible solution is a subset of elements satisfying some property.
We use $w(S)$ as a shorthand notation for vector $(\sum_{i\in S}w_{i1},\ldots, \sum_{i\in S}w_{id})$.
Our objective is to find a feasible solution $S$ such that
$\Exp[\mu(w(S)]$ is maximized.

From now on, $x$ and $k$ denote $d$-dimensional vectors
and $kx$ (or $k\cdot x$) denotes the inner product of $k$ and $x$.
As before, we assume $\mu(x)\in [0,1]$ for all $x\geq 0$
and $\lim_{|x|\rightarrow+\infty}\mu(x)=0$, where $|x|=\max(x_{1},\ldots,x_{d})$,
Our algorithm is almost the same as in the one dimension case and we briefly sketch it here.
We first notice that expected utilities decompose for exponential utility functions, i.e.,
$\Exp[\phi^{k\cdot w(S)}]= \prod_{i\in S}\Exp[\phi^{k\cdot w_{i}}]$.
Then, we attempt to
$\epsilon$-approximate the utility function $\mu(x)$
by a short exponential sum $\sum_{|k|\leq N} c_{k}  \phi_{k}^{kx}$ (there are $O(N^{d})$ terms).
If this can be done, $\Exp[\phi^{k\cdot w(S)}]$ can be approximated by $\sum_{|k|\leq N}c_{k}\Exp[\phi^{k\cdot w(S)}]$.
Using the same argument as in Theorem~\ref{thm:mainthm},
we can show that there is a polynomial time algorithm that can find a feasible solution $S$ with $\Exp[\mu(w(S))]\geq OPT-\epsilon$
for any $\epsilon>0$, provided that a pseudopolynomial algorithm exists for the exact version of the deterministic problem.

To approximate the utility function $\mu(x)$,
we need the multidimensional Fourier series expansion of a function
$f: \mathbb{C}^{d}\rightarrow \mathbb{C}$ (assuming $f$
is $2\pi$-periodic in each axis):
$
f(x) \sim \sum_{k\in \mathbb{Z}^{d}} c_{k}  e^{ikx}
$
where
$c_{k} = {1 \over (2 \pi)^d} \int_{x\in [-\pi,\pi]^{d}} f(x) e^{-ikx}\, dx $.
The {\em rectangular partial sum} is defined to be
$$
S_{N} f(x) =\sum_{|k_{1}|\leq N}\ldots\sum_{|k_{d}|\leq N} c_{k}  e^{ikx}.
$$
It is known that the rectangular partial sum  $S_{N}f(x)$ converges uniformly to $f(x)$ in $[-\pi,\pi]^{d}$
for many function classes as $n$ tends to infinity.
In fact, a generalization of Theorem~\ref{thm:fourier} to $[-\pi,\pi]^{d}$ also holds~\cite{alimov92multiple}:
If $f$ satisfies the  $\alpha$-H\"{o}lder condition, then
$$
|f(x)-(S_Nf)(x)|\le O\Bigl( {|f|_{C^{0,\alpha}}\ln^{d} N\over N^\alpha}\Bigr) \ \ \ \text{for }x\in [-\pi,\pi]^{d}.
$$
Now, we have an algorithm \algo\ that can approximate a function in a bounded domain.
It is also straightforward to extend \utidecomp\ to the multidimensional case. Hence,
we can $\epsilon$-approximate $\mu$ by a short exponential sum in $[0,+\infty)^{d}$, thereby proving
the multidimensional generalization of Theorem~\ref{thm:mainholder}.
Let us consider an application of our result.

\topic{Stochastic Multidimensional Knapsack (Revisited)}:
We consider the case where the components of each weight vector can be correlated.
Note that the utility function $\chi_{2}$ corresponding to this problem is the two dimensional threshold function:
$\chi_{2}(x,y)=1$ if $x\leq 1$ and $y\leq 1$; $\chi_{2}(x,y)=0$ otherwise.
As in the one dimensional case, we need to consider a continuous version $\tchi_{2}$ of $\chi_{2}$ (see Figure~\ref{fig_utility}(3)).
By the result in this section and a generalization of Lemma~\ref{lm:sanwich} to higher dimension, we can get the following.
\begin{theorem}
\label{thm:multidimknapsack}
For any constants $d\in \mathbb{N}$, $\epsilon>0$ and $0\leq \gamma\leq 1$,
there is a polynomial time algorithm for finding a set $S$ of items
such that the total profit $v(S)$ is $1-\epsilon$ factor of the optimum and
$\Prob(\bigwedge_{j=1}^{d} (\sum_{i\in S} w_{ij}\leq 1+\epsilon)) \geq (1-\epsilon)\gamma$.
\end{theorem}

\section{A Few Remarks}
\label{subsec:discussion}

\topic{Convergence of Fourier series}:
The convergence of the Fourier series of a function
is a classic topic in harmonic analysis.
Whether the Fourier series converges to the given function and
the rate of the convergence typically depends on a variety of
smoothness condition of the function.
We refer the readers to \cite{stein2003fourier}
for a more comprehensive treatment of this topic.
We note that we could obtain a smoother version of $\chi$ (e.g., see Figure~\ref{fig_utility}(2)),
instead of the piecewise linear $\tchi$,
and then use Theorem~\ref{thm:fourier} to obtain a better bound for $L$.
This would result in an even better running time.
Our choice is simply for the ease of presentation.

\topic{Discontinuous utility functions}:
If the utility function $\uti$ is discontinuous, e.g., the threshold function,
then the partial Fourier series behaves poorly around the discontinuity (this is known as the {\em Gibbs phenomenon}).
However, informally speaking, as the number of Fourier terms increases, the poorly-behaved strip around the
edge becomes narrower.
Therefore, if the majority of the probability mass of our solution lies outside the strip,
we can still guarantee a good approximation of the expected utility.
There are also techniques to reduce the effects of the Gibbs phenomenon (See e.g., \cite{gottlieb1997gibbs}).
However, the techniques are not sufficient to handle discontinuous functions.
We note that very recently,  Daskalakis et al. \cite{daskalakis2014polynomial} obtained a true additive PTAS
(instead of a bi-criterion additive PTAS) 
for a closely related problem, called {\em the fault tolerant storage problem},
under certain technical 
assumptions.~\footnote{
In the fault tolerant storage problem, we are given $n$ real numbers 
$0<p_1\leq \ldots\leq p_n<1$, and an addition number $0<\theta<1$.
Our goal is to partition $1$ into $n$ positive values $x_1,\ldots, x_n$ (i.e., $\sum_{i=1}^n x_i=1$),
such that $\Pr[\sum_{i=1}^n X_i\geq \theta]$ is maximized, where $X_i$
is the Bernoulli random variable which takes value $x_i$ with probability $p_i$.
In order to obtained an additive PTAS, Daskalakis et al. \cite{daskalakis2014polynomial} assumed
that all $p_i$s are bounded below by a constant.  
}
It is not clear how to use their technique to obtain a true additive PTAS for 
our expected utility maximization problem.
We leave this problem as an interesting open problem.

\section{Conclusion}
We study the problem of maximizing expected utility for
several stochastic combinatorial problems, such as shortest path, spanning tree and knapsack, and several classes of 
utility functions.
A key ingredient in our algorithm is to decompose the utility function into a short exponential sum, using
the Fourier series decomposition.
Our general approximation framework may be useful for other stochastic optimization problems.
We leave the problems of obtaining a true additive PTAS,
or nontrivial multiplicative approximation factors for $\classA$
as interesting open problems.

\section{Acknowledgments}
We would like to thank Evdokia Nikolova for providing an extended version of \cite{nikolova2006stochastic}
and many helpful discussions.
We also would like to thank Chandra Chekuri for pointing to us the work \cite{bhalgat10}
and Anand Bhalgat for some clarifications of the same work.

\bibliographystyle{abbrv}
\bibliography{paperbib}

\appendix

\section{Computing $\Exp[\phi^{w_{e}}]$}
\label{app:approxe}

If $X$ is a random variable, then the
{\em characteristic function} of $X$ is defined as
$$G(z) = \Exp[e^{izX}].$$
We can see $\Exp[\phi^{w_{e}}]$ is nothing but
the value of the characteristic function of $w_{e}$ evaluated at $-i\ln \phi$
(here $\ln$ is the complex logarithm function).
For many important distributions, including negative binomial,
Poisson, exponential, Gaussian, Chi-square and Gamma, a closed-form characteristic function is known.
See \cite{oberhettinger1973fourier} for a more comprehensive list.

\begin{example}
Consider the Poisson distributed $w_{e}$ with mean $\lambda$, i.e.,
$\Prob(w_{e}=k)=\lambda^{k}e^{-\lambda}/k!\,$.
Its characteristic function is known to be
$
G(z)= \! e^{\lambda(e^{iz}-1)}.
$
Therefore,
$$
\Exp[\phi^{w_{e}}]= G(-i\ln \phi)= \! e^{\lambda(\phi-1)}.
$$
\end{example}
\begin{example}
For Gaussian distribution $N(\mu, \sigma^{2})$, we know its characteristic function is
$
G(z)=e^{iz\mu - \frac{1}{2}\sigma^2z^2}.
$
Therefore,
$$
\Exp[\phi^{w_{e}}]= G(-i\ln \phi)=\phi^{u+\frac{1}{2}\sigma^{2}\ln \phi}.
$$
\end{example}

For some continuous distributions, no closed-form characteristic function is known
and we need proper numerical approximation method.

If the support of the distribution is bounded, we can use for example Gauss-Legendre quadrature \cite{Book_numerical}.
If the support is infinite, we can truncate the distribution and approximate the integral over the remaining finite interval;
Generally speaking a quadrature method approximates $\int_{a}^{b} f(x) \d x$ by a linear sum $\sum_{i=1}^k c_i f(x_i)$
where $c_i$ and $x_i$ are some constants independent of the function $f$.
A typical practice is to use  {\em composite rule}, that is to
partition $[a,b]$ into $N$ subintervals and approximate the integral using some quadrature formula over each subinterval.
For the example of Gauss-Laguerre quadrature, assuming continuity of the $2k$th derivative of $f(x)$ for some constant $k$,
if we partition $[a,b]$ into $M$ subintervals and apply Gauss-Legendre quadrature of degree $k$ to each subinterval,
the approximation error is
$$
\mathsf{Error}={(b-a)^{2k+1}\over M^{2k}}  {(k!)^4\over (2k+1)[(2k)!]^3} f^{(2k)}(\xi)
$$
where $\xi$ is some point in $(a,b)$ \cite[pp.116]{Book_numerical}.
Let $\Delta={b-a\over M}$.
If we treat $k$ as a constant, the behavior of the error (in terms of $\Delta$) is
$\mathsf{Error}(\Delta)=O(\Delta^{2k}\max_{\xi} f^{(2k)}(\xi))$.
Therefore, if the support and $\max_{\xi} f^{(2k)}(\xi)$ are bounded by a polynomial, we can approximate the integral, in polynomial time,
such that the error is $O(1/n^{\beta})$ for any fixed integer $\beta$.

The next lemma shows that we do not lose too much even though we can only get an
approximation of $\Exp[\phi^{w_{e}}]$.

\begin{lemma}
Suppose in Theorem~\ref{thm:algofortmu}, we can only compute an approximate value of $\Exp[\phi_{i}^{w_{e}}]$,
denoted by $E_{e,i}$, for each $e$ and $i$, such that $|\Exp[\phi_{i}^{w_{e}}]-E_{e,i}|\leq O(n^{-\beta})$ for some positive integer $\beta$.
Denote $E(S)=\sum_{k=1}^{L} c_{k} \prod_{e\in S}E_{e,i}$.
For any solution $S$, we have that
$$
|\Exp[\tmu(\len(S))]- E(S) | \leq O(n^{1-\beta}).
$$
\end{lemma}
\begin{proof}
We need the following simple result (see \cite{li2010ranking} for a proof):
$a_1,\ldots,a_n$ and $e_1,\ldots,e_n$ are complex numbers such that
$|a_i|\leq 1$ and
$|e_i|\leq n^{-\beta}$ for all $i$ and some $\beta>1$. Then, we have
$$
\Bigl|\prod_{i=1}^n (a_i+e_i)-\prod_{i=1}^n E_i \Bigr| \leq O(n^{1-\beta}).
$$
Since $|\phi_{i}|\leq 1$, we can see that
$$|\Exp[\phi_{i}^{w_{e}}]|=|\int_{x\geq 0} \phi_{i}^{x}p_{e}(x)\d x| \leq 1.$$
The lemma simply follows by applying the above result and
noticing that $L$ and all $c_{k}$s are constants. \qed
\end{proof}

We can show that Theorem~\ref{thm:mainthm} still holds
even though we only have the
approximations of the $\Exp[\phi^{w_{e}}]$ values.
The proof is straightforward and omitted.

\end{document}